%% file: note_on_CoVaR.tex
\documentclass[12pt,a4paper,notitlepage]{article} 
\pdfoutput=1
\usepackage[latin1]{inputenc}
\usepackage{amsmath}
\usepackage{amsthm}
\usepackage{amsfonts}
\usepackage{amssymb}
\usepackage{makeidx}
\usepackage{graphicx}
\usepackage{enumerate}
\usepackage{placeins} 

\graphicspath{{graphics/}}

\usepackage[numbers]{natbib}
\usepackage{stackrel}
\usepackage{colonequals} 
\usepackage{url}
        
\usepackage[plainpages=false]{hyperref}
\usepackage{hyperref}
\hypersetup{  colorlinks=true,%
              citecolor=black,%
              filecolor=black,%
              linkcolor=black,%
              urlcolor=black,%
}
\DeclareMathOperator{\es}{ES}  
\DeclareMathOperator{\var}{VaR}  
\DeclareMathOperator{\covar}{CoVaR}
\DeclareMathOperator{\ABcovar}{CoVaR^{=}}
\DeclareMathOperator{\ABcovara}{CoVaR_{\alpha}^{=}}
\DeclareMathOperator{\ABcovaraa}{CoVaR_{\alpha,\alpha}^{=}}
\DeclareMathOperator{\ABcovarab}{CoVaR_{\alpha, \beta}^{=}}
\DeclareMathOperator{\covara}{CoVaR_{\alpha}}
\DeclareMathOperator{\covaraa}{CoVaR_{\alpha,\alpha}}
\DeclareMathOperator{\covarab}{CoVaR_{\alpha, \beta}}
\DeclareMathOperator{\coes}{CoES}
\DeclareMathOperator{\coesa}{CoES_{\alpha}}

\DeclareMathOperator{\coesab}{CoES_{\alpha, \beta}}
\DeclareMathOperator{\ABcoesa}{CoES_{\alpha}^{=}}

\DeclareMathOperator{\ABcoesab}{CoES_{\alpha, \beta}^{=}}
\DeclareMathOperator{\ABcoes}{CoES^{=}}
\DeclareMathOperator{\dcovar}{\Delta CoVaR^{=}}
\DeclareMathOperator{\dcovarab}{\Delta CoVaR^{=}_{\alpha, \beta}}
\DeclareMathOperator{\dcovara}{\Delta CoVaR^{=}_{\alpha}}
\DeclareMathOperator{\dmcovar}{\Delta^{\mathrm{med}} CoVaR^{=}}
\DeclareMathOperator{\dmcovarab}{\Delta^{\mathrm{med}} CoVaR^{=}_{\alpha, \beta}}
\DeclareMathOperator{\dmcovara}{\Delta^{\mathrm{med}} CoVaR^{=}_{\alpha}}

\DeclareMathOperator{\med}{med}

\DeclareMathOperator{\corr}{corr}
\DeclareMathOperator{\variance}{var}
\DeclareMathOperator{\MES}{MES}

\newcommand{\cubr}[1]{\{ #1 \}}
\newcommand{\disteq}{\stackrel{\mathrm{d}}{=}}
\newcommand{\eps}{\varepsilon}
\newcommand{\E}{\mathrm{E}}
\newcommand{\Ecal}{\mathcal{E}}
\newcommand{\enquote}[1]{\lq\lq{}{#1}\rq\rq{}}
\newcommand{\FX}{F_X}
\newcommand{\fX}{f_X}
\newcommand{\FXY}{F_{X,Y}}
\newcommand{\fXY}{f_{X,Y}}
\newcommand{\FY}{F_Y}
\newcommand{\ginv}{^\leftarrow}
\newcommand{\inv}{^{-1}}
\newcommand{\Law}{\mathcal{L}}
\newcommand{\muX}{\mu_X}
\newcommand{\muY}{\mu_Y}
\newcommand{\mylefteqn}{\hspace{2em}&\hspace{-2em}}

\newcommand{\Ncal}{\mathcal{N}}
\newcommand{\order}{\mathrel{\preceq}}
\newcommand{\ordersm}{\mathrel{\preceq}_{\mathrm{sm}}}
\newcommand{\R}{\mathbb{R}}

\newcommand{\robrfl}[1]{\left({#1}\right)}
\newcommand{\sigmaX}{\sigma_X}

\newcommand{\sigmaY}{\sigma_Y}
\newcommand{\Xtilde}{\widetilde{X}}
\newcommand{\Ytilde}{\widetilde{Y}}
\newcommand{\tr}{^{\top}}
\newcommand{\unif}{\mathrm{unif}}

\renewcommand{\d}{\text{d}}
\renewcommand{\P}{\mathrm{P}}
\theoremstyle{plain}
\newtheorem{theorem}{Theorem}[section]
\newtheorem{corollary}[theorem]{Corollary} 
\newtheorem{definition}[theorem]{Definition} 
\theoremstyle{remark}
\newtheorem{remark}[theorem]{Remark} 

\newcommand{\draftonly}[1]{\relax}

\newcommand{\footnoteremember}[2]{\footnote{#2}\newcounter{#1}\setcounter{#1}{\value{footnote}}}
\newcommand{\footnoterecall}[1]{\footnotemark[\value{#1}]}

\author{
Georg Mainik%
\footnoteremember{foot:RiskLab}{RiskLab, Department of Mathematics, ETH Zürich, Rämistrasse 101, 8092 Zürich, Switzerland} 
\footnote{\texttt{georg.mainik@math.ethz.ch}}
\and
Eric Schaanning%
\footnoterecall{foot:RiskLab}
\footnote{\texttt{eric.schaanning@math.ethz.ch}}
}
\title{On dependence consistency of $\covar$ and some other systemic risk measures}

\begin{document}
\bibliographystyle{abbrvnat} 

\maketitle

\draftonly{
\tableofcontents
\listoffigures
}

\begin{abstract}
\input{abstract}
\end{abstract}


\input{introduction}
\input{basics} 
\input{general}
\input{examples}
\input{conclusions}
\input{acknowledgements}


\bibliography{references}

\end{document}

%% file: abstract.tex
This paper is dedicated to the consistency of systemic risk measures 
with respect to stochastic dependence.  
It compares two alternative notions of Conditional Value-at-Risk 
($\covar$) available in the current literature. 
These notions are both based on the conditional distribution 
of a random variable $Y$ given a stress event for a random variable $X$,
but they use different types of stress events.
%
%
We derive representations of these alternative $\covar$ notions in 
terms of copulas, study their general dependence consistency and compare 
their performance in several stochastic models. 
Our central finding is that conditioning on $X\ge\var_\alpha(X)$ 
gives a much better response to dependence between $X$ and $Y$ 
than conditioning on $X=\var_\alpha(X)$. 
We prove general results that relate the dependence consistency of $\covar$ 
using conditioning on $X\ge\var_\alpha(X)$ to well established results on 
concordance ordering of multivariate distributions or their copulas. 
These results also apply to some other systemic risk measures, such 
as the Marginal Expected Shortfall ($\mathrm{MES}$) and the Systemic Impact Index ($\mathrm{SII}$). 
We provide counterexamples showing that $\covar$ based on the stress event 
$X = \var_\alpha(X)$ is not dependence consistent. 
In particular, if $(X,Y)$ is bivariate normal, then $\covar$ based on $X=\var_\alpha(X)$ is not an increasing function of the correlation parameter. 
Similar issues arise in the bivariate $t$ model and in the model with 
$t$ margins and a Gumbel copula. 
In all these cases, $\covar$ based on $X\ge\var_\alpha(X)$ is an increasing 
function of the dependence parameter. 
%
%
%
%
%

%% file: introduction.tex
\section{Introduction}
The present paper studies the notion of Conditional  Value-at-Risk ($\covar$) 
introduced by Adrian and Brunnermeier \cite{covar08} 
as a dependence adjusted version of 
Value-at-Risk ($\var$). 
The general idea behind $\covar$ is to use the conditional distribution 
of a random variable $Y$ representing a particular financial institution 
(or the entire financial system) given that another institution, represented by 
a random variable $X$, is in stress.
$\covar$ represents one of the major threads in the current regulatory 
and scientific discussion of systemic risk, which significantly 
intensified after the recent financial crisis. 
The current discussion on systemic risk measurement is far from being 
concluded, and the competing methodologies  
are still under development. 
In addition to systemic risk measures 
\cite[cf.][]{covar08,covar11,Girardi/Erguen:2011,Goodhart/Segoviano:2008,Acharya/Pedersen/Philippon/Richardson:2010,Zhou:2010,Huang/Zhou/Zhu:2011},  
related topics include the structure of interbank networks, e.g., 
\cite{Boss/Elsinger/Summer/Thurner:2004,Cont/Moussa/Santos:2012}, 
models explaining how systemic risk is created, e.g., 
\cite{Choi/Douady:2012,Ibragimov/Walden:2007}, 
and attribution of systemic risk charges within a financial system, 
as discussed in \cite{Tarashev/Borio/Tsatsaronis:2010,Staum:2012}.
\par
Our contribution addresses  
the consistency of systemic risk measures 
with respect to the dependence in the underlying stochastic model. 
In the case of $\covar$ we give a strong indication for the choice of 
the stress event for the conditioning random variable $X$.
%
%
%
There are two alternative definitions of $\covar$ in the current 
literature. 
The original definition in \cite{covar08,covar09,covar11} is 
derived from the conditional distribution of $Y$ given that  
$X=\var_\alpha(X)$. 
The second one uses conditioning on $X\ge\var_\alpha(X)$.
This modification 
was proposed by Girardi and Ergün \cite{Girardi/Erguen:2011} to 
improve the compatibility of $\covar$ with non-parametric estimation methods. 
For similar reasons, such as continuity and better compatibility with discrete 
distributions, conditioning on $X\ge\var_\alpha(X)$ was also favoured 
by Klyman \cite{Klyman:2011} 
for both $\covar$ and Conditional Expected Shortfall ($\coes$). 
Finally, it is remarkable that most competitors of $\covar$ 
\cite[cf.][]{Goodhart/Segoviano:2008,Acharya/Pedersen/Philippon/Richardson:2010,Zhou:2010,Huang/Zhou/Zhu:2011} use conditioning on 
$X\ge\var_\alpha(X)$ as well. This approach goes in line with 
the general concept of stress scenarios discussed in~\cite{Balkema/Embrechts:2007}.
\par
Our results show that conditioning on $X\ge\var_\alpha(X)$ 
has great advantages for dependence modelling.
We prove that this modification of $\covar$ 
makes it response consistently to dependence parameters in 
many important stochastic models, 
whereas the original definition of $\covar$ fails to do so. 
The counterexamples even include the bivariate Gaussian model, 
where the original $\covar$ is decreasing with respect to 
the correlation $\rho:=\corr(X,Y)$ for $\rho>1/\sqrt{2}$. 
Thus, $\covar$ based on $\cubr{X=\var_\alpha(X)}$ fails to detect 
systemic risk when it is most pronounced; and we also found this kind of 
inconsistency in other examples. 
On the other hand, our findings for the modified $\covar$ 
relate its dependence consistency to 
concordance ordering of multivariate distributions or related copulas. 
This may explain the comparative results in \cite{Gauthier/Lehar/Souissi:2012}, 
where $\ABcovar$ stood somewhat apart from its competitors.
Moreover, it gives the the modified notion of $\covar$ 
a solid mathematical basis. 
\par
Besides $\covar$, we also discuss extensions to Conditional Expected Shortfall 
($\coes$). It turns out 
that the dependence inconsistency or dependence consistency 
of the alternative $\covar$ notions 
is propagated to the corresponding definitions of $\coes$.
The dependence consistency results for 
$\covar$ and $\coes$ based on the stress scenario $X\ge\var_\alpha(X)$ 
also apply to the Marginal Expected Shortfall ($\mathrm{\MES}$)
defined in~\cite{Acharya/Pedersen/Philippon/Richardson:2010} 
and to the Systemic Impact Index ($\mathrm{SII}$) introduced in~\cite{Zhou:2010}.  
\par
The paper is organized as follows. 
Basic notation and alternative definitions of $\covar$ and $\coes$  
are given in Section~\ref{sec:basics}.  
In Section~\ref{sec:general} we present the general mathematical 
results, including representations of $\covar$ in terms of copulas 
and consistency of the modified $\covar$ or $\coes$ with respect to 
dependence characteristics. 
Section~\ref{sec:examples} contains a detailed comparison of the 
original and the modified $\covar$ in three different models: 
the bivariate normal, the bivariate $t$ distribution, and 
a bivariate distribution with $t$ margins and a Gumbel copula. 
Conclusions are stated in Section~\ref{sec:conclusions}. 
%
%
%

%% file: basics.tex
\section{Basic definitions and properties} \label{sec:basics}
Let $X$ and $Y$ be random variables representing the profits and losses of 
two financial institutions, such as banks. 
Focusing on risks, let $X$ and 
$Y$ be \emph{random loss} variables, so that positive values of $X$ and $Y$ 
represent losses, whereas the gains are represented by negative values. 
\par
The issues of contagion and systemic stability raise questions for the 
joint probability distribution of $X$ and $Y$:
\[
F_{X,Y}(x,y) := \P (X \le x, Y \le y). 
\]
The corresponding marginal distributions will be denoted by $\FX$ and $\FY$.
Provided a method to quantify the loss or gain of the entire financial system,
$F_{X,Y}$ can also represent the joint loss distribution of a bank $X$ and the  
system $Y$. 
\par
In the current banking regulation framework 
(Basel II and the so-called Basel 2.5),  
the calculation of risk capital is based on measuring risk of each 
institution separately, with \emph{Value-at-Risk} ($\var$) as a risk measure.
The Value-at-Risk of a random loss $X$ at the confidence 
level $\alpha\in(0,1)$ is the $\alpha$-quantile of the loss 
distribution $\FX$ \citep[cf.][Definition 2.10]{qrm}. That is, 
\[
\var_\alpha(X) = \FX\ginv(\alpha) 
\] 
where $\FX\ginv(y):= \inf\{ x \in \R : \FX(x) \ge y\}$ 
is the \emph{generalized inverse} of $\FX$. 
The most common values of $\alpha$ are $0.95$ and $0.99$. 
\par
For continuous and strictly increasing $\FX$ the generalized inverse 
$\FX\ginv$ coincides with the inverse function $\FX\inv$ of $\FX$. 
In this case one has $\var_\alpha(X) = \FX\inv(\alpha)$ for $\alpha\in(0,1)$. 
For a thorough discussion of generalized inverse functions we refer to 
\cite{inverses}.
\par
In the present paper we discuss two alternative approaches to adjust 
$\var$ to dependence between $X$ and $Y$. 
This is achieved by conditioning the distribution of $Y$ on a stress 
scenario for $X$. 
These two notions appear in the recent literature under the name 
\emph{Conditional Value-at-Risk} ($\covar$), 
but they use different kinds of stress scenarios. 
The original notion was introduced by Adrian and Brunnermeier 
\cite{covar08,covar09,covar11} and will henceforth be denoted by
$\ABcovar$.
The alternative definition was proposed by Girardi and Ergün  \cite{Girardi/Erguen:2011}. We denote it by $\covar$. 
\begin{definition}\label{def:1}
\begin{align*}
\ABcovarab (Y|X) 
&:= 
\var_\beta(Y|X=\var_\alpha(X));\\ 
\covarab(Y|X) 
&:= 
\var_\beta(Y|X\ge\var_\alpha(X)). 
\end{align*}
\end{definition}
\par
The computation of $\ABcovar$ requires 
the knowledge of $F_{Y|X=\var_\alpha(X)}$. 
If $\FXY$ has a density $\fXY$, then $\fX(x)=\int_{-\infty}^\infty \fXY(x,y) dy$ is a density of $\FX$, and 
\[
F_{Y|X=\var_\alpha(X)}(y) = \frac{\int_{-\infty}^y \fXY(\var_\alpha(X), t) \,\d t}{\fX(\var_\alpha(X))},
\]
provided that $\fX(\var_\alpha(X)) > 0$. 
In some models, such as elliptical distributions, $F_{Y|X=\var_\alpha(X)}$ is 
known explicitly. In general, however, computation of $F_{Y|X=\var_\alpha(X)}$ 
requires numerical integration. 
\par
Conditioning on $X\ge \var_\alpha(X)$ is less technical.
The definition of $\var_\alpha(X)$ implies that 
$\P(X\ge\var_\alpha(X)) \ge 1-\alpha$, 
so that elementary conditional probabilities are well defined. 
In particular, if $\FX$ is continuous, then 
\[
F_{Y|X\ge\var_\alpha(X)} (y) = \frac{\P(Y\le y, X\ge\var_\alpha(X))}{1-\alpha}. 
\]
Moreover, conditioning on events with positive probabilities is 
advantageous in statistical applications, including model fitting and 
backtesting.
This is the major reason why 
the original notion of $\ABcovar$ was modified to $\covar$ 
in~\cite{Girardi/Erguen:2011}. 
\par
A straightforward extension from $\covar$ to \emph{Conditional Expected Shortfall} ($\coes$) is based on the representation
$\es_\beta(Y)=\frac{1}{1-\beta}\int_\beta^1 \var_t(Y) dt$. 
\begin{definition}\label{def:2}
\begin{align}
\coesab (Y|X) 
&\colonequals  \label{eq:019}
\frac{1}{1 - \beta} \int_{\beta}^{1} \covar_{\alpha, t} (Y|X) \d t,\\
\ABcoesab (Y|X) 
&\colonequals \label{eq:019a}
\frac{1}{1 - \beta} \int_{\beta}^{1} \ABcovar_{\alpha, t} (Y|X) \d t.
\end{align}
\end{definition}
\begin{remark} \label{rem:3}
\begin{enumerate}[(a)]
\item
In precise mathematical terms, $\ABcovarab$ and $\covarab$ 
are the $\beta$-quantiles of the conditional distributions
$F_{Y|X=\var_\alpha(X)}$ and $F_{Y|X\ge\var_\alpha(X)}$:
\begin{align*}
\ABcovarab (Y|X) 
&= 
F_{Y|X=\var_\alpha(X)}\ginv(\beta);\\
\covarab (Y|X) 
&= 
F_{Y|X\ge\var_\alpha(X)}\ginv(\beta).
\end{align*}
\item \label{item:rem.3.b}
In \cite{covar08,covar11,ambro,Girardi/Erguen:2011}, 
the authors work with a common confidence level for $X$ and $Y$, 
i.e., in the special case $\alpha=\beta$. 
Similarly to the notation used there, 
we will omit $\beta$ if $\beta=\alpha$ and 
write $\covar_{\alpha}$ instead of $\covar_{\alpha,\alpha}$ if 
it does not lead to confusion.
However, the 
definition of $\coes$ needs separate confidence levels for $X$ and $Y$ 
in the integrand $\covar_{\alpha,t}(Y|X)$.  
\item
Since $\coesab (Y|X)= \es_\beta(Z)$ for 
a random variable $Z \sim F_{Y | X \geq \var_{\alpha} (X) }$, 
the \emph{coherence} of $\es$ in the sense of~\cite{Artzner/Delbaen/Eber/Heath:1999} is inherited by $\coesab$ for all $\alpha,\beta\in(0,1)$.
The central point here is subadditivity, which is understood as 
\[
\coesab(Y + Y'|X) \le \coesab(Y|X) + \coesab(Y'|X)
\]
for any random variables $(Y,Y',X)$ defined on the same probability space. 
%
%
\item
In \cite{covar08,covar11}, $\coes$ is defined as $\E[Y|Y\ge\ABcovaraa(Y|X)]$. 
Note that this definition replaces the stress scenario 
$\cubr{X=\var_\alpha(X)}$ by $\cubr{Y\ge\ABcovaraa(Y|X)}$, 
which is not related to $X$ directly. Compared to $\ABcoesab(Y|X)$, 
this definition is quite unnatural. Moreover, it does not guarantee 
coherence, which is the central property of Expected Shortfall.
\item \label{item:rem.3.e}
The notion of \emph{Marginal Expected Shortfall} ($\MES$) introduced in \cite{Acharya/Pedersen/Philippon/Richardson:2010} is closely related to $\covar$ and $\coes$. It is defined as 
\[
\MES_\alpha(Y|X) := \E [Y|X\ge\var_\alpha(X)]
\]
where $X:=\sum_{i=1}^dY_i$ is the financial system and $Y:=Y_i$ for some $i$ is 
an institution. 
The idea behind $\MES$ is to quantify the insurance premia corresponding 
to bail-outs which become necessary when the entire financial system is 
close to a collapse. 
The major economic difference between 
$\MES$ and $\covar$ is the role of $X$ and $Y$. 
With $\MES$, 
the conditioning random variable $X$ is the system, 
and the target random variable $Y$ is a part of the system. 
In the original work on $\covar$,  
$Y$ is the system, and $X$ is a part of it. 
\par
On the mathematical level, $\MES$ and $\covar$ or $\coes$ are quite close 
to each other. It is easy to see that 
\begin{align*}
\MES_\alpha(Y|X) 
&= 
\int_0^1 F_{Y|X\ge \var_\alpha (X)} \ginv(t) d t 
=
\int_0^1 \covar_{\alpha,t}(Y|X) d t.
\end{align*}
In view of~\eqref{eq:019}, one could 
also write $\MES_\alpha(Y|X)=\coes_{\alpha,0}(Y|X)$. 
\item 
In \cite{Klyman:2011}, $\covarab$ and $\coesab$ in the sense of 
Definitions~\ref{def:1} and~\ref{def:2} are called $\mathrm{DistVaR}$ 
and $\mathrm{DistES}$. Besides the different naming, the definitions 
are essentially the same, and these notions are also compared to 
$\ABcovar$ and $\ABcoes$. However, the comparison in~\cite{Klyman:2011} is 
concentrated on general representations, 
compatibility with discrete, e.g., empirical, distributions, 
and the behaviour in the bivariate Black-Scholes model. 
As far as we are aware, a study of consistency with respect to dependence 
parameters has been missing so far.
\end{enumerate}
\end{remark}
\par
The introduction of $\ABcovar$ in \cite{covar08} aims not at 
$\ABcovar$ itself, but at the contribution of a particular financial 
institution to the systemic risk. 
In~\cite{covar08}, $\ABcovar$ is used to construct 
a \emph{risk contribution measure} that should quantify 
how a stress situation for an institution $X$ affects the system (or another 
institution) $Y$. 
In~\cite{covar08}, the authors propose $\frac{\ABcovarab(Y|X)}{\var_\beta(Y)}-1$ 
as a systemic risk indicator. 
In \cite{covar09}, the systemic risk measure is modified to 
\begin{equation} \label{eq:021} 
\dcovarab(Y|X) := \ABcovarab (Y) - \var_\beta(Y).
\end{equation}
In \cite{covar11}, the centring term $\var_\beta(Y)$ representing 
the risk of $Y$ in an unstressed state is replaced by the conditional 
$\var$ of $Y$ given that $X$ is equal to its median:
\begin{equation} \label{eq:022}
\dmcovarab(Y|X) := \ABcovarab (Y|X) - \var_\beta(Y|X=\med(X)) 
\end{equation}
to remedy some inconsistencies observed in a comparison of $\ABcovar$ 
across different models. 
\par
Unfortunately, the centring in~\eqref{eq:021} is not the only reason why 
$\dcovar$ can give a biased view of dependence between $X$ 
and $Y$. 
The results presented below demonstrate that there is a more fundamental 
issue that cannot be solved by modifying $\dcovar$ to $\dmcovar$ 
or taking any other centring term.
The primary deficiency of $\dcovar$ is that the 
underlying stress scenario $X=\var_\alpha(X)$ 
is too selective and over-optimistic. 
If, for instance, $\FX$ is continuous, then $\P(X=\var_\alpha(X)) = 0$, so that 
this particular event actually never occurs. 
Generally speaking, the ability of $\ABcovar$, $\dcovar$, or $\dmcovar$ 
to describe the influence of 
$X$ on $Y$ strongly depends on how well 
$F_{Y|X=\var_\alpha(X)}$ approximates $F_{Y|X=x}$ for $x\ge\var_\alpha(X)$. 
As shown in Section~\ref{sec:examples}, this approximation fails even in 
very basic models, and it typically underestimates the contagion from  $X$ to $Y$. 
%
%
%
%
%

%% file: general.tex
\section{General results} \label{sec:general}
We begin with representations of $\ABcovar$ and $\covar$ 
in terms of copulas. 
It is well known that any bivariate distribution function $F_{X,Y}$ 
admits the decomposition
\begin{align} \label{eq:008} 
F_{X,Y}(x,y) = C(\FX(x), \FY(y))  
\end{align}
where $C$ is a probability distribution function on $(0,1)^2$ with uniform 
margins (cf. \cite{Sklar:1959,Joe:1997}).   
That is, there exist random variables $U,V\sim\unif(0,1)$ such that 
$C(u,v) = \P(U\le u, V\le v)$. 
The function $C$ is called a \emph{copula} of $F_{X,Y}$. 
If both $\FX$ and $\FY$ are continuous, then $C$ is uniquely 
determined by 
$C(u,v) = F_{X,Y}(\FX\ginv(u), \FY\ginv(v))$. 
\par
The decomposition~\eqref{eq:008} yields the following representation of 
$\ABcovar$ and $\covar$.
\begin{theorem}\label{thm:1}
Let $(U,V)\sim C$ where $C$ is a copula of $\FXY$. 
If $\FX$ is continuous, then
\begin{enumerate}[(a)]
\item \label{item:thm.1.a}
$\ABcovarab (Y|X) = \FY\ginv\robrfl{F_{V|U=\alpha}\ginv(\beta)}$, 
\item \label{item:thm.1.b}
$\covarab (Y|X) = \FY\ginv\robrfl{F_{V|U\ge\alpha}\ginv(\beta)}$, 
and $F_{V|U\ge\alpha}(v) = \frac{v-C(\alpha,v)}{1-\alpha}$.
\end{enumerate}
\end{theorem}
\begin{proof}
Part~(\ref{item:thm.1.a}).
It is well known that $(\FY\ginv(U),\FX\ginv(V))\sim \FXY$, and hence  
\[
F_{Y | X = \var_\alpha(X)} (y)
=
\P(\FY\ginv(V) \le y | \FX\ginv(U) = \FX\ginv(\alpha)).  
\]
The functions $\FY$ and $\FY\ginv$ are non-decreasing and satisfy 
$v\le\FY(\FY\ginv(v))$ and $\FY\ginv(\FY(y)) \le y$ for all $v\in(0,1)$ and $y\in\R$. 
This implies that $\FY\ginv(V)\le y$ is equivalent to $V\le \FY(y)$.  
Moreover, continuity of $\FX$ implies that $\FX\ginv$ is strictly increasing, 
so that $\FX\ginv(U) = \FX\ginv(\alpha)$ 
is equivalent to $U=\alpha$. This yields
\[
F_{Y | X = \var_\alpha(X)} (y)
=
\P(V\le \FY(y) | U=\alpha)
=
F_{V|U=\alpha}(\FY(y)), 
\]
and the result follows from the chain rule for the generalized inverse.
\par
Part(\ref{item:thm.1.b}).
Analogously to Part~(\ref{item:thm.1.a}), one obtains that 
\[
F_{Y | X \ge \var_\alpha(X)} (y)
=
\P(V\le \FY(y) | U\ge\alpha)
=
F_{V|U\ge\alpha}(\FY(y)),
\]
and hence $\covarab(Y|X) = \FY\ginv(F_{V|U\ge \alpha}(\beta))$. 
Since $(U,V)\sim C$ and the margins of $C$ are uniform, we obtain that 
\begin{align*}
F_{V|U\ge \alpha}(v) 
&= \frac{P(V\le v, U\ge\alpha)}{P(U\ge\alpha)} =  
\frac{v- C(\alpha,v)}{1-\alpha}. 
\end{align*}
\end{proof}
Theorem~\ref{thm:1}(\ref{item:thm.1.a}) provides a link between the 
ordering of $\covar$ and the notion of \emph{concordance ordering}. 
\begin{definition}
\citep[cf.][Definition 3.8.1]{Mueller/Stoyan:2002}
Let $(X,Y)$ and $(X',Y')$ be bivariate random vectors with $\FX=F_{X'}$ and 
$\FY = F_{Y'}$.
Then $(X,Y)$ is smaller than $(X',Y')$ in concordance order ($(X,Y)\order(X',Y')$ or, equivalently, $\FXY \order F_{X',Y'}$) if 
\[
\forall x,y\in\R
\quad
\P(X\le x,Y\le y) \le \P(X'\le x,Y'\le y). 
\]
\end{definition}
\begin{remark} \label{rem:2}
The following equivalent characterizations of $(X,Y)\order(X',Y')$ 
will be used in in the sequel: 
\begin{enumerate}[(a)]
\item 
$P(X>x, Y>y) \le P(X'>x, Y'>y)$ for all $x,y\in\R$; 
\item 
$C \order C'$ for the copulas of $\FXY$ and $F_{X',Y'}$ if the 
margins are continuous;
\item \label{item:rem.2.b}
$\E f(X,Y) \le \E f(X',Y')$ for all supermodular functions $f:\R^2\to\R$, 
i.e., for all $f$ satisfying 
\[
f(x+\eps,y+\delta) + f(x,y) \ge f(x+\eps ,y) + f(x, y+\delta) 
\]
for all $x,y\in\R$ and all $\eps,\delta>0$. This order relation is 
called \emph{supermodular ordering ($\ordersm$)}. 
\end{enumerate}
For proofs and further alternative characterizations we refer to 
\cite[Theorem 3.8.2]{Mueller/Stoyan:2002}.
\end{remark}
\par
The central theoretical result of the present paper is the following. 
\begin{theorem} \label{thm:3}
Let $(X,Y)$ and $(X',Y')$ be bivariate random vectors with copulas $C$ and $C'$, respectively, and assume that $\FY = F_{Y'}$. 
\begin{enumerate}[(a)]
\item \label{item:thm.3.a}
If $\FX$ and $F_{X'}$ are continuous, then $C \order C'$ implies
\begin{equation} \label{eq:016}
\forall \alpha,\beta\in(0,1)
\quad
\covarab(Y|X) \le \covarab(Y'|X').
\end{equation}
\item \label{item:thm.3.b}
If $\FX$, $F_{X'}$, $\FY$, and $F_{Y'}$ are continuous, 
then~\eqref{eq:016} implies $C \order C'$.  
\end{enumerate}
\end{theorem}
\begin{remark}
Note that Theorem~\ref{thm:3} does not need $\FX=F_{X'}$. The only assumption 
on the conditioning random variables $X$ and $X'$ is that they are 
continuously distributed.
\end{remark}
\begin{proof}
Part~(\ref{item:thm.3.a}). 
Let $(U,V)\sim C$ and $(U',V')\sim C'$. 
As $\FX\ginv$ and $\FY\ginv$ are non-decreasing, 
Theorem~\ref{thm:1}(\ref{item:thm.1.b}) reduces the problem to 
\begin{equation} \label{eq:017}
\forall \alpha,\beta\in(0,1)
\quad
F_{V|U \ge \alpha}\ginv (\beta) \le F_{V'|U' \ge \alpha}\ginv (\beta).
\end{equation}
Moreover, it is well known that for any distribution functions $G$ and $H$ the 
ordering $G\ginv(y) \le H \ginv (y)$ for all $y\in(0,1)$ is equivalent to 
$G(x) \ge H(x)$ for all $x\in\R$. Thus it suffices to show that 
\begin{equation*}
\forall \alpha,v \in(0,1)
\quad
F_{V|U\ge\alpha}(v) \ge F_{V'|U'\ge\alpha} (v).
\end{equation*}
The representation of $F_{V|U\alpha}(v)$ in Theorem~\ref{thm:1}(\ref{item:thm.1.b}) reduces this to $C(\alpha,v) \le C'(\alpha,v)$ for all $\alpha,v$, which is 
precisely $C \order C'$.
\par
Part~(\ref{item:thm.3.b}). Combining ~\eqref{eq:016} with Theorem~\ref{thm:1}, 
one obtains 
\begin{equation} \label{eq:018}
\forall \alpha,\beta 
\quad
\FY\ginv(F_{V|U\ge\alpha}\ginv (\beta)) \le \FY\ginv(F_{V'|U'\ge\alpha}\ginv (\beta)).
\end{equation}
As $\FY$ is continuous, $\FY\ginv$ is strictly increasing. 
Therefore~\eqref{eq:018} implies \eqref{eq:017}, 
which is equivalent to $C \order C'$.
\end{proof}
\par
Theorem~\ref{thm:3} can be applied to various stochastic models. 
We start with elliptical distributions. 
This model class includes such important examples 
as the multivariate Gaussian and the multivariate $t$ distributions. 
Since $\covarab(X|Y)$ 
considers two random variables and multivariate ellipticity implies 
bivariate ellipticity for all bivariate sub-vectors, 
we restrict the consideration to the bivariate case. 
\par
A bivariate random vector $(X_1,X_2)$ is \emph{elliptically distributed} if 
\[
(X,Y)\tr \disteq \mu\tr + R A W\tr
\]
where $\mu=(\muX,\muY)\in\R^2$ and $A\in\R^{2\times2}$ are constant, 
$W=(W_1,W_2)$ is uniformly distributed on the 
Euclidean unit sphere $\cubr{x\in\R^2: \|x\|_2 = 1}$, and $R$ is a non-negative 
random variable independent of $W$. 
If $\E R<\infty$, then $\muX=\E X$ and $\muY=\E Y$. 
The \emph{ellipticity matrix} $\Sigma:=A\tr A$ is unique except for 
a multiplicative factor. 
The covariance matrix of $(X,Y)$ is defined if and only if $\E R^2<\infty$, 
and this matrix is always equal to $c\Sigma$ for some constant $c>0$. 
Thus, rescaling $R$ and $A$, one can always achieve that 
\begin{equation}
\label{eq:020}
\Sigma = \left(\begin{array}{cc} \sigmaX^2 & \sigmaX\sigmaY\rho \\ \sigmaX\sigmaY\rho &\sigmaY^2 \end{array}\right)
\end{equation}
where, if defined, $\sigmaX=\variance(X)$, $\sigmaY=\variance(Y)$, and 
$\rho=\corr(X,Y)$. In the following we will always assume this standardization 
of $\Sigma$ and denote the bivariate elliptical distribution with 
location parameter $\mu=(\muX,\muY)$ and ellipticity matrix $\Sigma$ by
$\Ecal(\mu,\Sigma,R)$. 
\par
If $(X,Y) \sim\Ecal(\mu,\Sigma,R)$ with continuous marginal distributions, 
then the copula $C$ of $(X,Y)$ is uniquely determined. 
Copulas of this type are called \emph{elliptical copulas}. 
The invariance of copulas under increasing 
marginal transforms implies that $C$ depends only on the parameter $\rho$ of
$\Sigma$ and on the distribution of $R$. Thus $\rho$ is the natural dependence 
parameter for a bivariate elliptical copula $C$, 
whereas the distribution of $R$ specifies 
the type of the copula, such as Gaussian or $t$. 
We will call elliptical copulas $C$ and $C'$ of \emph{same type} if 
the corresponding elliptical distributions have identical radial 
parts $R\disteq R'$.
\par
The following theorem states monotonicity of $\covar$ with respect 
to the dependence parameter $\rho$ if $(X,Y)$ is elliptically distributed 
or has an elliptical copula. 
In particular, it applies to 
bivariate Gaussian or bivariate $t$ distributions, 
and also to bivariate distributions with Gaussian or $t$ copulas. 
\par
\begin{theorem}\label{thm:4}
\begin{enumerate}[(a)]
\item \label{item:thm.4.a}
Let $(X,Y)\sim\Ecal(\mu,\Sigma,R)$ 
and $(X',Y')\sim\Ecal(\mu',\Sigma',R)$ with 
continuous $\FX$ and $F_{X'}$. 
If $\muY\le\mu_{Y'}$ and $\sigmaY=\sigma_{Y'}$, 
then $\rho \le \rho'$ implies \eqref{eq:016}.  
\item \label{item:thm.4.c}
Let $(X,Y)\sim\Ecal(\mu,\Sigma,R)$ 
and $(X',Y')\sim\Ecal(\mu',\Sigma',R)$ with 
continuous $\FX$ and $F_{X'}$. 
If $\muY\le\mu_{Y'}$ and $\sigmaY\le\sigma_{Y'}$, then $\rho \le \rho'$ implies 
\[
\forall \alpha\in(0,1) \, \forall\beta\in[\beta_0,1)
\quad
\covarab(Y|X) \le \covarab(Y'|X')
\]
with $\beta_0:=\frac{1/2 - C(\alpha,1/2)}{1-\alpha}$ where $C$ is the copula 
of $(X,Y)$.
\item \label{item:thm.4.b}
Let $\FXY$ and $F_{X',Y'}$ have elliptical copulas of same type
with dependence parameters $\rho$ and $\rho'$, respectively. 
If $\FX$ and $F_{X'}$ are continuous and $\FY (y) \ge F_{Y'}(y)$ 
for all $y\in\R$, then $\rho \le \rho'$ implies~\eqref{eq:016}. 
\end{enumerate}
\end{theorem}
\begin{remark}
\begin{enumerate}[(a)]
\item
The assumption $\FY \ge F_{Y'}$ obviously includes the case of identical 
margins $\FY=F_{Y'}$, which is the natural setting for studying the response of 
$\covar$ to dependence parameters. 
\item
It is easy to see that the lower bound $\beta_0$ in 
Theorem~\ref{thm:4}(\ref{item:thm.4.c}) is decreasing in $\rho$. 
In particular, one has $\beta_0\le 1/2$ for $\rho\ge 0$. 
This guarantees that $\covarab(Y|X) \le \covarab(Y'|X')$ for 
$\alpha,\beta\in[1/2,1)$, which is fully sufficient for 
assessing dependence between rare events. 
\end{enumerate}
\end{remark}
\begin{proof}[Proof of Theorem~\ref{thm:4}]
Part~(\ref{item:thm.4.a}). 
It is obvious that $\covarab(c+Y|X) = c+\covarab(Y|X)$. 
Hence, as $\muY\le\mu_{Y'}$, it suffices to consider $\muY=\mu_{Y'}$, 
so that we have $\FY=F_{Y'}$.  
Since the case $\sigmaY=0$ is trivial, 
we only need to consider $\sigmaY>0$.
\par
The continuity of $\FX$ yields $\sigmaX>0$, and as $(X,Y)$ is elliptically 
distributed, we have $Y\disteq \frac{\sigmaY}{\sigmaX}X$. 
Hence $\FY$ is continuous as well, and therefore the copulas $C$ and $C'$ 
of $(X,Y)$ and $(X',Y')$ are uniquely defined. 
\par
According to Theorem~\ref{thm:3}(\ref{item:thm.3.a}), it suffices 
to show that $\rho<\rho'$ implies $C \order C'$. 
This is equivalent to 
$\Ecal(0,0,\Gamma(\rho),R) \order \Ecal(0,0,\Gamma(\rho'),R)$ 
for $\rho \le \rho'$ 
and $\Gamma(\rho)=\left(\begin{array}{cc} 1 & \rho \\ \rho & 1 \end{array}\right)$.
This ordering result is proven in~\cite{Cambanis/Simons:1982}. 
In the bivariate Gaussian case it is also known as Slepian's inequality 
\citep[cf.][Theorem 5.1.7]{Tong:1990}. 
\par
Part~(\ref{item:thm.4.c}). 
Without loss of generality we can assume that $\muY=\mu_{Y'}$ and $\sigmaY>0$. 
Part~(\ref{item:thm.4.a}) gives us $C\order C'$ and hence~\eqref{eq:017}.
Moreover, $\sigmaY \le \sigma_{Y'}$ implies that 
$\FY\ginv(t) \le F_{Y'}\ginv(t)$ for $t\in[1/2,1)$. 
Hence, according to Theorem~\ref{thm:1}(\ref{item:thm.1.b}), it suffices 
to verify that $F_{V|U\ge\alpha}\ginv(\beta) \ge 1/2$. 
This inequality is equivalent to $\beta\ge F_{V|U\ge\alpha}(1/2) = \beta_0$. 
\par
Part~(\ref{item:thm.4.b}). According to Part~(\ref{item:thm.4.a}), we have 
$C\order C'$ and hence~\eqref{eq:017}.
Since $\FY(y)\ge F_{Y'}(y)$ for all $y\in\R$ is equivalent to 
$\FY(y)\ginv(t)\le F_{Y'}\ginv(t)$, Theorem~\ref{thm:1}(\ref{item:thm.1.b}) yields
\[
\covarab(Y|X) \le F_{Y'}\ginv(F_{V|U\ge\alpha}\ginv(\beta)) \le \covarab(Y'|X').
\]
\end{proof}
\par
A very popular copula model is the \emph{Gumbel copula}. 
In the bivariate case it is defined as  
\begin{equation} \label{eq:023}
C_\theta(u,v) = \exp\robrfl{-\robrfl{(-\log u )^\theta + (- \log v )^\theta}^{1/\theta}}. 
\end{equation}
The dependence parameter $\theta$ assumes values in $[1,\infty]$, whereas 
$\theta=1$ and $\theta=\infty$ refer to $C_1(u,v):=uv$ (independence copula) 
and $C_\infty(u,v):=\min(u,v)$ (comonotonicity copula).
As shown in~\cite{Hu/Wei:2002}, $\theta \le \theta'$ 
implies $C_\theta \ordersm C_{\theta'}$ and hence $C_\theta \order C_{\theta'}$ 
(cf.\ Remark~\ref{rem:2}(\ref{item:rem.2.b})). 
This immediately yields the following analogue 
of Theorem~\ref{thm:4}(\ref{item:thm.4.b}).
\begin{corollary} \label{cor:1}
Let $(X,Y)$ and $(X',Y')$ have Gumbel copulas with dependence parameters $\theta$ and $\theta'$, respectively. 
If $\FX$ and $F_{X'}$ are continuous and $\FY(y)\ge F_{Y'}(y)$ for all $y\in\R$, then $\theta \le \theta'$ implies~\eqref{eq:016}.  
\end{corollary}
\begin{remark}
Corollary~\ref{cor:1} also holds for \emph{Galambos copulas} 
with dependence parameters $\theta \le \theta'$; 
see \cite{Hu/Wei:2002} for $C_\theta \ordersm C_{\theta'}$ in this case.
\end{remark}
The monotonicity of $\coesab(X,Y)$ with respect to dependence 
parameters follows from the integral representation~\eqref{eq:019}. 
\begin{corollary} \label{cor:2} 
Suppose that $\E |Y|$ and $\E |Y'|$ are finite. 
\begin{enumerate}[(a)]
\item \label{item:cor.2.a}
If $(X,Y)$ and $(X',Y')$ satisfy the assumptions of 
Theorem~\ref{thm:4}(\ref{item:thm.4.a}) or (\ref{item:thm.4.b}), or those of Corollary~\ref{cor:1}, then
\begin{equation} \label{eq:024}
\forall \alpha,\beta \in(0,1)
\quad
\coesab(Y|X) \le \coesab(Y'|X').
\end{equation}
\item \label{item:cor.2.b}
If $(X,Y)$ and $(X',Y')$ satisfy the assumptions of 
Theorem~\ref{thm:4}(\ref{item:thm.4.c}), then 
\[
\forall \alpha \in(0,1) \, \forall \beta \in[\beta_0,1)
\quad
\coesab(Y|X) \le \coesab(Y'|X').
\]
with $\beta_0=\frac{1/2 - C(\alpha,1/2)}{1-\alpha}$.
\end{enumerate}
\end{corollary}
We conclude this section by relating the results obtained here 
to another systemic risk measure.
\begin{remark}
\begin{enumerate}[(a)]
\item
Corollary~\ref{cor:2}(\ref{item:cor.2.a}) also applies to the Marginal Expected Shortfall from \cite{Acharya/Pedersen/Philippon/Richardson:2010}.
Setting  $\beta=0$ in~\eqref{eq:024} and applying 
Remark~\ref{rem:3}(\ref{item:rem.3.e}), 
one obtains $\MES_\alpha(Y|X)\le\MES_\alpha(Y'|X')$ for all $\alpha$.
\item
In~\cite{Zhou:2010}, the Systemic Impact Index ($\mathrm{SII}$) of an institution $Y_i$ 
is defined as 
\begin{align*}
\mathrm{SII}_i(\alpha) 
&:= 
\E \robrfl{\sum_{j=1}^d 1\cubr{Y_j \ge \var_\alpha(Y_j)} \bigg| Y_i \ge \var_\alpha(Y_i)}\\
&=
1+ \sum_{j\ne i} \P(Y_j \ge \var_\alpha(Y_j) | Y_i\ge \var_\alpha(Y_i)). 
\end{align*}
It is easy to see that~\eqref{eq:016} is equivalent to 
\[
\P(Y>\var_\beta(Y)|X>\var_\beta(X)) \le \P(Y'>\var_\beta(Y')|X'>\var_\beta(X'))
\]
for all $\alpha,\beta$. 
Thus, for $Y= Y_j$ and $X=Y_i$, 
the assumptions of Theorems~\ref{thm:3}(\ref{item:thm.3.a}) 
and \ref{thm:4}
also imply dependence consistency of the single 
conditional default probabilities 
$\P(Y_j \ge \var_\alpha(Y_j) | Y_i\ge \var_\alpha(Y_i))$.
\end{enumerate}
\end{remark}
%
%
%

%% file: examples.tex
\section{Examples} \label{sec:examples}
In this section we compare $\covar$ and $\ABcovar$ in three different 
models: the bivariate Gaussian, the bivariate $t$, and the bivariate 
distribution with a Gumbel copula and $t$ margins. 
\subsection{The bivariate Gaussian distribution} \label{subsec:gaussian}
It is well known that the bivariate Gaussian distribution is elliptical.
Hence Theorem~\ref{thm:4}(\ref{item:thm.4.a}) 
guarantees that $\covar$ is an increasing function of the 
correlation parameter $\rho$. 
Moreover, $\ABcovar$ can be calculated explicitly in this case, so that 
it is particularly easy to compare $\covar$ with $\ABcovar$. 
\subsubsection*{Computation of $\ABcovar$}
Let $(X,Y)\sim\Ncal(\mu,\Sigma)$ with mean vector $\mu=(\muX,\muY)$ and 
covariance matrix $\Sigma$ as in~\eqref{eq:020}. 
As for all bivariate elliptical models, the dependence between $X$ and $Y$ is 
fully described by the correlation parameter $\rho$.
An appealing property of the bivariate normal distribution 
is the interpretation as a 
linear model. Indeed, $(X,Y)\sim\Ncal(\mu,\Sigma)$ is equivalent to 
\begin{align} \label{eq:002}
\frac{Y-\muY}{\sigmaY} = \rho \frac{X-\muX}{\sigmaX} + \sqrt{1-\rho^2} Z,
\end{align}
where $X\sim\Ncal(\muX,\sigma_X^2)$ and $Z\sim\Ncal(0,1)$, independent of $X$. 
\par
Due to $X\sim\Ncal(\muX,\sigmaX^2)$ 
we have $\var_\alpha(X) = \muX + \sigmaX \Phi\inv(\alpha)$, 
where $\Phi$ is the distribution function of 
$\Ncal(0,1)$. 
Substituting $X=\var_\alpha(X)$ in~\eqref{eq:002}, one obtains  
\[
Y=\muY + \sigmaY \robrfl{\rho\Phi\inv(\alpha) + \sqrt{1-\rho^2} Z}.
\]
This shows that 
$\Law(Y|X=\var_\alpha(X)) =  \Ncal(\tilde{\mu},\tilde{\sigma}^2)$ 
with $\tilde{\mu}=\muY + \sigmaY\rho\Phi\inv(\alpha)$ and 
$\tilde{\sigma}=\sigmaY\sqrt{1-\rho^2}$. 
Hence we obtain that 
\begin{align}
\ABcovarab (Y|X) 
&= \nonumber
\var_\beta (Y | X=\var_\alpha(X)) = \tilde{\mu} + \tilde{\sigma}\Phi\inv(\beta) \\
&= \label{eq:006}
\muY + \sigmaY \robrfl{  \rho \Phi\inv(\alpha) + \Phi\inv(\beta) \sqrt{1 - \rho^2} }. 
\end{align}
\subsubsection*{Computation of $\covar$}
To compute $\covar$, we use the copula representation from 
Theorem~\ref{thm:1}(\ref{item:thm.1.b}). 
From $Y\sim\Ncal(\muY,\sigma^2)$ one obtains that
$\FY\inv(v) = \muY + \sigmaY\Phi\inv(v)$ for $v\in(0,1)$. 
Moreover, the copula of $(X,Y)\sim\Ncal(\mu,\Sigma)$ 
is the Gauss copula $C_\rho$
with dependence parameter $\rho$. For $\rho=0$ it is the independence 
copula, $C_0(u,v)=uv$, and for $\rho\ne0$ it has the following representation:
\begin{align}
C_\rho(u,v)
&=\nonumber 
\FXY(\FX\ginv(u),\FY\ginv(v)) \\
&= \label{eq:003}
\int_{- \infty}^{\Phi^{-1}(v)} \int_{- \infty}^{\Phi^{-1}(u)} 
\frac{1}{2 \pi \sqrt{1 - \rho^2}} 
\exp \robrfl{\frac{-(s_1^2 - 2\rho s_1 s_2 + s_2^2)}{2(1-\rho^2)}}
\d s_2 \d s_1.
\end{align}
Applying Theorem~\ref{thm:1}(\ref{item:thm.1.b}), we obtain 
\[
\covarab(Y) = \muY + \sigmaY\Phi\inv(F_{V|U \ge \alpha}\inv(\beta))
\] 
where $F_{V|U \ge \alpha}(v) = \frac{v - C_\rho(\alpha,v)}{1-\alpha}$. 
The values of $\covar$ can be obtained by numerical integration 
of~\eqref{eq:003} and numerical inversion of the function $F_{V|U \ge \alpha}(v)$. 
\par
An alternative method to compute $\covar$ is the numerical computation 
and inversion of the function
\begin{align} \label{eq:004}
F_{Y|X\ge\var_\alpha(X)} (t) 
=
\frac{1}{1-\alpha} 
\int_{-\infty}^{t}\int_{\var_\alpha(X)}^\infty \fXY (x,y) \, \d x \,\d y,
\end{align}
where $\fXY$ is the joint density of $X$ and $Y$. 
Depending on the application, each method has its advantages. 
Whilst 
\eqref{eq:004} is more direct and hence faster for numerically 
tractable $\fXY$, 
the conditional copula values obtained in~\eqref{eq:003} 
can be re-used with other marginal distributions. 
\subsubsection*{Monotonicity in $\rho$}
As bivariate Gaussian distributions are elliptical, 
Theorem~\ref{thm:4}(\ref{item:thm.4.a}) guarantees that 
$\covar$ is always increasing in $\rho$. 
However, this is not the case for $\ABcovar$. 
Partial differentiation of~\eqref{eq:006} in $\rho$ yields 
\begin{align} \label{eq:005}
\partial_\rho \ABcovarab(Y|X) 
= 
\sigmaY \robrfl{ \Phi\inv (\alpha)  - \frac{\rho \Phi\inv (\beta) }{\sqrt{1 - \rho^2}} },
\end{align} 
which is positive if $\Phi\inv(\alpha) \sqrt{1-\rho^2} > \rho \Phi\inv(\beta)$ 
and negative if $\Phi\inv(\alpha) \sqrt{1-\rho^2} < \rho \Phi\inv(\beta)$. 
Besides the degenerate case 
$\alpha=\beta=1/2$ with constant $\ABcovarab$, 
there are 4 cases depending on the signs of $\Phi\inv(\alpha)$ 
and $\Phi\inv(\beta)$:
\begin{enumerate}[(i)]
\item 
If $\alpha \ge 1/2$ and $\beta \ge 1/2$, then $\ABcovarab(Y|X)$ is increasing in $\rho$ for $\rho < \rho_0:= \frac{|\Phi\inv(\alpha)|}{\sqrt{(\Phi\inv(\alpha))^2 + (\Phi\inv(\beta))^2}}$ and decreasing for $\rho> \rho_0$. 
\item 
If $\alpha\ge 1/2$ and $\beta<1/2$, then $\ABcovarab(Y|X)$ is increasing in $\rho$ for $\rho > - \rho_0$ and decreasing for $\rho < -\rho_0$.
\item 
If $\alpha<1/2$ and $\beta\ge 1/2$, then $\ABcovarab(Y|X)$ is increasing in $\rho$ for $\rho < - \rho_0$ and decreasing for $\rho > -\rho_0$. 
\item 
If $\alpha<1/2$ and $\beta<1/2$, then $\ABcovarab(Y|X)$ is increasing in $\rho$ for $\rho > \rho_0$ and decreasing for $\rho < \rho_0$. 
\end{enumerate}
Thus $\ABcovar$ is monotonic with respect to $\rho$ only in degenerate cases. 
In particular, in the most important case $\alpha,\beta\in(1/2, 1)$, $\ABcovar$ 
is decreasing for $\rho>\rho_0$, which means that $\ABcovar$ 
fails to detect dependence where it is most pronounced. 
In the special case $\alpha=\beta$, the critical threshold $\rho_0$ 
is always equal to $1/\sqrt{2}$. 
\par
\begin{figure} 
\centering
\includegraphics[width=0.8\textwidth]{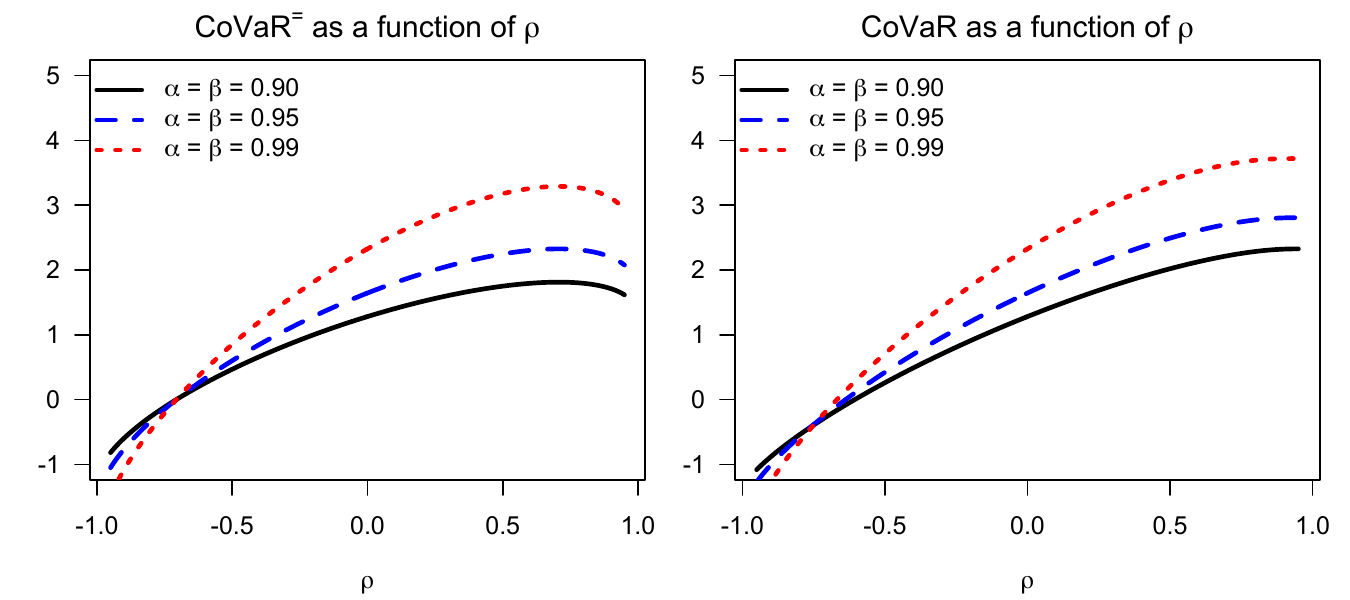}
\caption{$\ABcovara (Y|X)$ and $\covara(Y|X)$ (i.e., with $\beta=\alpha$) in the bivariate normal model as functions of $\rho$.}
\label{fig:normal-covar-scovar-rho}
\end{figure}
\par
A graphic illustration to this fact is given in 
Figure~\ref{fig:normal-covar-scovar-rho}, showing $\ABcovara(Y|X)$ 
and $\covara(Y|X)$
for $\rho \ge 0.2$ and $\alpha=\beta$ assuming values $0.90,$ $0.95$, 
or $0.99$. 
The short writing $\ABcovara$ refers to $\ABcovaraa$; 
analogously, $\covara$ denotes $\covaraa$.  
This notation was used in the original definitions of 
$\ABcovar$ and $\covar$, which were restricted to $\alpha=\beta$ 
(cf.\ Remark~\ref{rem:3}(\ref{item:rem.3.b})). 
For the sake of simplicity we set $\muY=0$ and $\sigmaY=1$. 
These parameters 
have no influence 
on the decreasing or increasing behaviour 
of $\covar$ or $\ABcovar$ as functions of $\rho$.
\subsubsection*{Normalized values of $\covar$ and $\ABcovar$}
\begin{figure}
\centering
\includegraphics[width=0.8\textwidth]{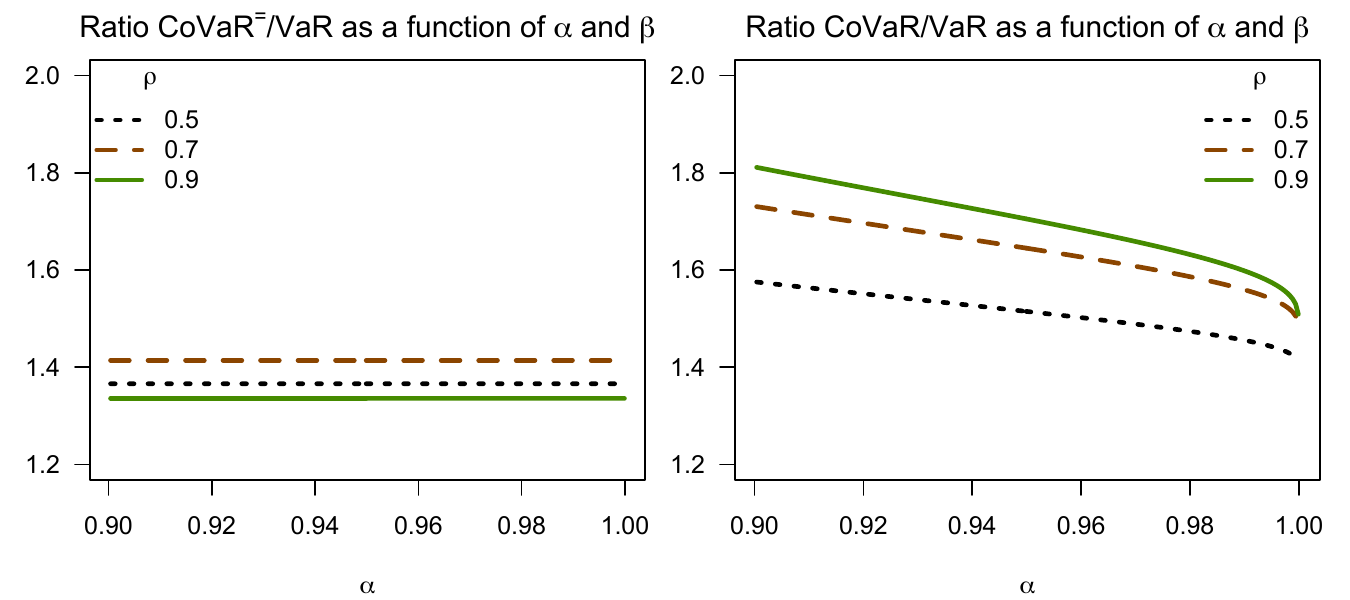}
\caption{Bivariate normal model with $\mu_Y=0$: Ordering of the ratios $\ABcovara(Y|X)/\var_\alpha(Y)$ and $\covara(Y|X)/\var_\alpha(Y)$ for different $\alpha$.} 
\label{fig:ratios-normal}
\end{figure}
\par
The relative impact of a stress event for $X$ on the institution $Y$  
can be quantified by the ratio $\ABcovarab(Y|X)/\var_\alpha(Y)$ 
or by $\covarab(Y|X)/\var_\alpha(Y)$.
A similar of systemic risk indicator was proposed in~\cite{covar08}. 
Figure~\ref{fig:ratios-normal} shows these ratios for $\alpha=\beta$ 
and $\mu=0$ as functions of $\alpha$.  
The different line types in the plots correspond to 
$\rho=0.5$, $0.7$, and $0.92$. 
The ratios $\ABcovara(Y|X)/\var_\alpha(Y)$ are constant, 
which is also easy to see from~\eqref{eq:006}. 
The interesting part here is the ordering of the lines for 
different $\rho$. 
In case of $\ABcovar$, the line for $\rho=0.7$ is above the 
two others, illustrating that the inconsistency issue 
is common to all $\alpha\in(1/2,1)$. 
The plot of $\covara(Y|X)/\var_\alpha(Y)$ 
shows correct ordering for all $\alpha$, 
as guaranteed by Theorem~\ref{thm:4}(\ref{item:thm.4.a}). 
Another observation one can make here is that 
$\covara(Y|X)/\var_\alpha(Y)$ 
is decreasing in $\alpha$. 
This, however, is a model property 
that seems to be related to the light tail of the normal distribution. 
In heavy-tailed models considered in Sections~\ref{subsec:ellipt-t} and~\ref{subsec:gumbel-t} 
the ratio $\covara(Y|X)/\var_\alpha(Y)$ is increasing in $\alpha$. 
\subsubsection*{Backtesting and violation rates}
The results above show that 
$\covar$ reflects the dependence between $X$ 
and $Y$ much more consistently than $\ABcovar$. 
An intuitive and very general explanation to this fact is that 
conditioning on $X\ge\var_\alpha(X)$ 
corresponds to a reasonable \enquote{what if} question, 
whereas conditioning on $X=\var_\alpha(X)$ does not.
Indeed, the scenario $\cubr{X\ge\var_\alpha(X)}$ includes all 
possible outcomes for $X$ if $X$ is stressed, whereas the 
scenario $\cubr{X=\var_\alpha(X)}$ selects only the most benign 
of them. 
\par
In backtesting of $\var$ one expects that $X$ exceeds $\var_\alpha(X)$ with 
probability not larger than $1-\alpha$. 
Abbreviating \enquote{Conditional $\var$}, the term $\ABcovarab$ 
suggests that $Y$ exceeds $\ABcovarab(Y|X)$ with 
conditional probability $1-\beta$ or less, given that $X$ is stressed.
The definition of $\covar$ understands stress of $X$ as 
$\cubr{X\ge\var_\alpha(X)}$, so that the expected violation rate for $\covarab$ 
under this stress scenario is equal to $1-\beta$. 
In contrast to that, $\ABcovar$ is designed to have the violation rate 
$1-\beta$ under the less natural and more optimistic scenario 
$\cubr{X=\var_\alpha(X)}$.
As a consequence, the violation rates for $\ABcovarab$ 
backtesting experiments based on the natural stress scenario 
$\cubr{X\ge\var_\alpha(X)}$ are significantly higher than $1-\beta$. 
\par
This issue is illustrated in Table~\ref{table:normal-table}. 
The underlying Monte Carlo experiment generates 
an i.i.d. sample $(X_i,Y_i)\sim \Ncal(0,\Sigma)$ 
for $i=1,\ldots,n$ and counts the joint exceedances 
$\cubr{Y_i\ge\ABcovarab(Y|X),  X_i\ge \var_\alpha(X)}$. 
The $\ABcovar$ violation rate for the stress scenario 
$\cubr{X\ge\var_\alpha(X)}$ is the ratio of the joint 
excess count and the count of the excesses $\cubr{X_i\ge\var_\alpha(X)}$. 
The violation rate for $\covar$ is obtained analogously from the number 
of joint exceedances $\cubr{Y_i\ge\covarab(Y|X), X_i\ge \var_\alpha(X)}$.
We chose $n=10^7$ and $\alpha,\beta$ being either $0.95$ or $0.99$. 
\par
It is remarkable that 
the violation rate for $\ABcovar$
increases with $\rho$. 
This demonstrates that the underestimation of risk by $\ABcovar$ is most 
pronounced in case of strong dependence and, hence, high systemic risk. 
\par
\begin{table}[ht]
\centering 
\begin{tabular}{c c c c c c} 
\hline\hline 
Bound & $ \rho = 0$ & $ \rho = 0.2 $  & $ \rho = 0.5 $ & $ \rho = 0.7 $ & $ \rho = 0.9 $ \\ [0.5ex] 
\hline 
$\covar_{0.95, 0.95}^{=}(Y|X)$ & 0.0503 & 0.0601 & 0.0857 & 0.1229 & 0.2520 \\ 
$\covar_{0.99, 0.99}^{=}(Y|X)$ & 0.0099 & 0.0124 & 0.0189 & 0.0292 & 0.0875 \\ 
$\covar_{0.95, 0,99}^{=}(Y|X)$ & 0.0101 & 0.0130 & 0.0213 & 0.0375 & 0.1224 \\ 
$\covar_{0.99, 0.95}^{=}(Y|X)$ & 0.0500 & 0.0588 & 0.0785 & 0.1045 & 0.2053 \\ [0.9ex]
$\covar_{0.95, 0.95}(Y|X)$ & 0.0503 & 0.0500 & 0.0503 & 0.0495 & 0.0499 \\ 
$\covar_{0.99, 0.99}(Y|X)$ & 0.0099 & 0.0101 & 0.0104 & 0.0099 & 0.0098 \\ 
$\covar_{0.95, 0.99}(Y|X)$ & 0.0101 & 0.0102 & 0.0102 & 0.0099 & 0.0098 \\ 
$\covar_{0.99, 0.95}(Y|X)$ & 0.0500 & 0.0507 & 0.0509 & 0.0501 & 0.0491 \\ 
\hline 
\end{tabular}
\caption{Violation rates in the bivariate normal case. Monte Carlo backtesting with $n=10^{7}$ and 
$\alpha,\beta\in\cubr{0.95,0.99}$%
} 
\label{table:normal-table} 
\end{table}
\par
A graphical illustration to this issue is given in 
Figure~\ref{fig:simul-normal-1} by bivariate normal samples from 
the simulation study described above. 
The horizontal lines mark the levels of $\ABcovara(Y|X)$ 
and $\covara(Y|X)$, and $\var_\alpha(Y)$. The vertical lines mark 
$\var_\alpha(X)$. 
The joint excess counts are 
the  numbers of points above the corresponding horizontal line and 
on the right hand side from the vertical line marking $\var_\alpha(X)$. 
The sample size is $n=2000$,  
which suffices to demonstrate how correlation changes the shape of the 
sample cloud and thus increases the number of the 
joint excesses $\cubr{Y_i\ge\ABcovara(Y|X), X_i\ge\var_\alpha(X)}$. 
\par
\begin{figure} 
\centering
\includegraphics[width=\textwidth]{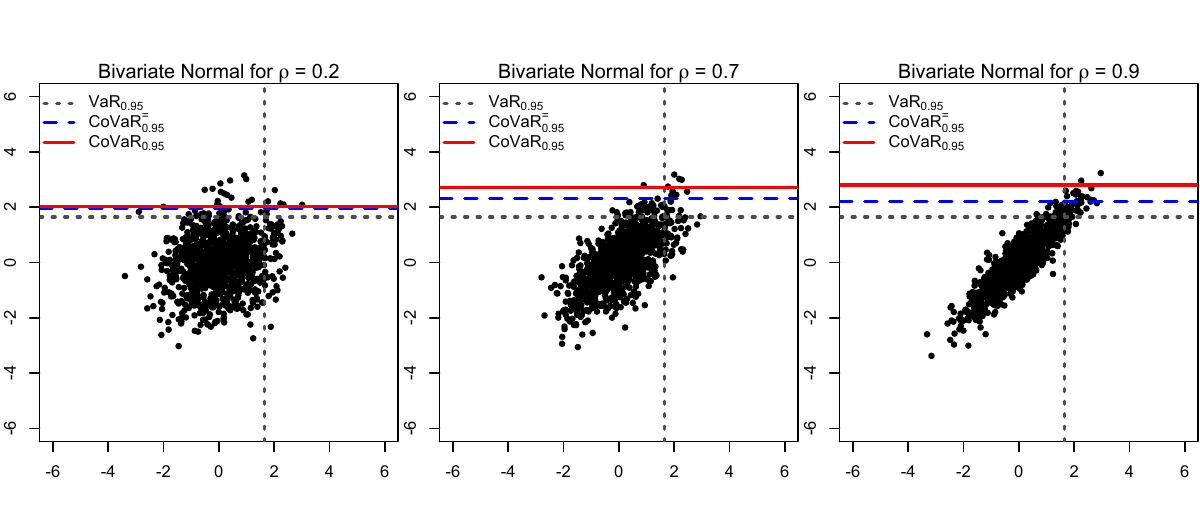}
\caption{Bivariate normal samples (size $n=2000$)
and the joint excess regions in the backtesting experiment for $\alpha=\beta=0.95$.}
\label{fig:simul-normal-1}
\end{figure}
\subsubsection*{Risk contribution measures $\dcovar$ and $\dmcovar$}
\par
\begin{figure} 
\centering
\includegraphics[width=0.8\textwidth]{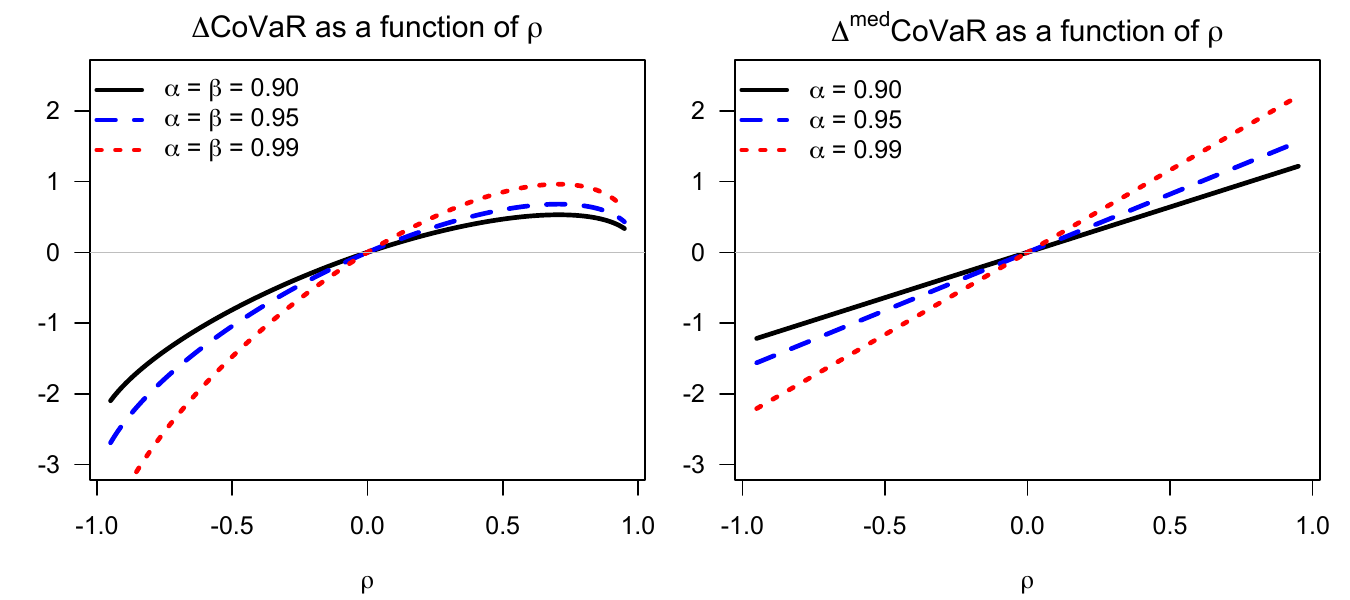}
\caption{$\dcovara$ and $\dmcovara$ as functions of $\rho$ in the bivariate normal model.}
\label{fig:Normal-dcovar-rho}
\end{figure}
\par 
As mentioned in Section~\ref{sec:basics}, \cite{covar08} aims not at 
$\ABcovar$ itself, but at the difference between $\ABcovar$ and some 
characteristic of an unstressed state. The two most common definitions 
of such a risk contribution measure 
are $\dcovar$ and $\dmcovar$ (see \eqref{eq:021} and~\eqref{eq:022}). 
In the bivariate normal case one has 
$\var_\beta (Y) = \muY + \sigmaY\Phi\inv(\beta)$, so that~\eqref{eq:006} 
yields
\[
\dcovarab(Y) = \sigmaY \left( \Phi\inv\ ( \alpha) \rho + \Phi\inv (\beta) \left(  \sqrt{1-\rho^2} -1 \right) \right).
\]
For $\alpha = \beta$ this simplifies to  
$\dcovara(Y) = \sigmaY \Phi\inv(\alpha)\robrfl{\rho + \sqrt{1-\rho^2} -1}$.
Regardless of $\alpha$ and $\beta$, $\dcovar$ inherits 
the non-monotonicity in $\rho$ from $\ABcovar$. 
An illustration to this issue is given in 
Figure~\ref{fig:Normal-dcovar-rho}, which shows 
plots of $\dcovar$ and $\dmcovar$ as functions of $\rho$ for $\alpha=\beta$.   
\par
At a first glance, $\dmcovar$ seems to be an improvement because 
it is increasing in $\rho$. In fact, $\dmcovar$ is even linear here. 
Due to $\med(X)=\muX$, \eqref{eq:002} yields 
$F_{Y|X=\med(X)}\ginv(\beta) = \muY + \sigmaY\sqrt{1-\rho^2}\Phi\inv(\beta)$. 
Applying~\eqref{eq:006}, one obtains that 
\begin{align}
\mylefteqn \nonumber
\dmcovarab (Y)\\
&= \nonumber
\muY + \sigmaY \left( \Phi\inv (\alpha ) \rho + \Phi\inv (\beta ) \sqrt{ 1 - \rho^2} \right) 
- 
\robrfl{\muY + \sigmaY \Phi\inv(\beta)  \sqrt{1-\rho^2}}\\
&=  \label{eq:007}
\sigmaY \Phi\inv(\alpha)\rho.
\end{align}
Thus, in the bivariate normal model, $\dmcovarab (Y|X)$ is linear 
with positive slope that depends on $\rho$ and $\alpha$, but not on $\beta$. 
In view of the linear structure~\eqref{eq:002} of the bivariate Gaussian 
model, this even appears reasonable.  
However, examples in 
Sections~\ref{subsec:ellipt-t} and~\ref{subsec:gumbel-t} 
show that $\dmcovar$ is not a monotonic function of dependence 
parameters in other models. 
Thus the applicability of $\dmcovar$ is restricted to linear models of 
type~\eqref{eq:002}, 
where it is superfluous because it carries quite the same information as 
the correlation parameter $\rho$ 
or the linear regression 
parameter from the classical Capital Asset Pricing Model
(the so-called CAPM-$\beta$), which is equal to $\rho \sigmaY/\sigmaX$ in the 
present setting.
\subsubsection*{Extension from $\covar$ to $\coes$}
\begin{figure} 
\centering
\includegraphics[width=0.8\textwidth]{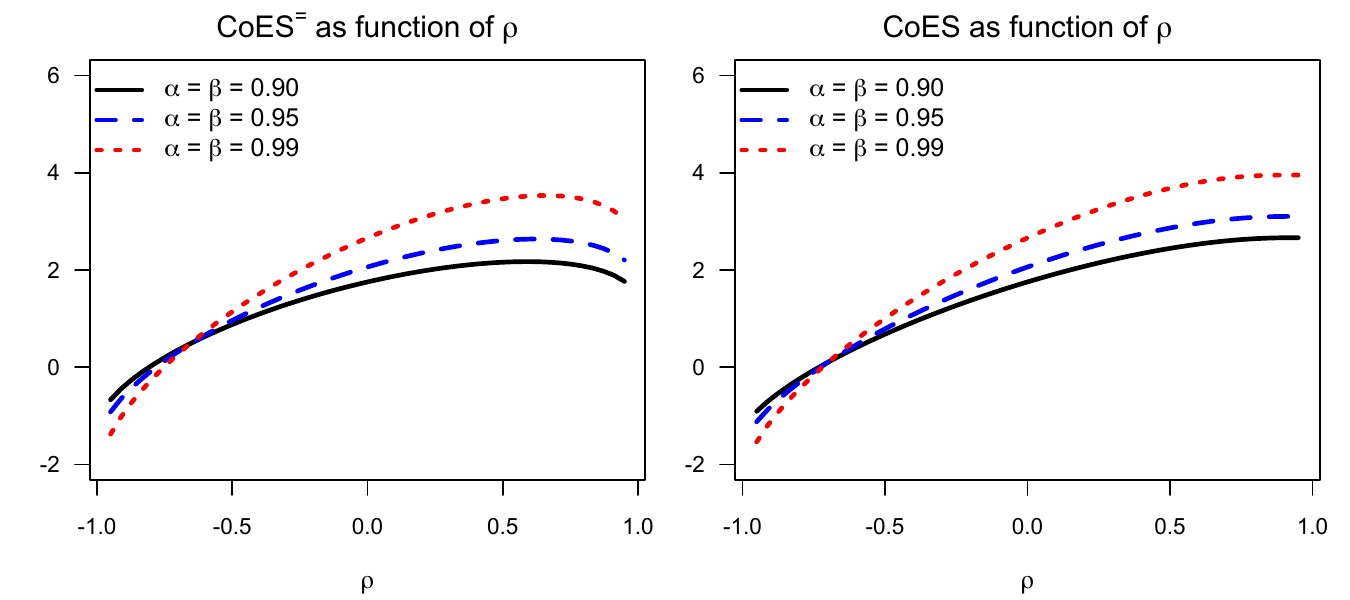}
\caption{$\ABcoesa (Y|X)$ and $\coesa(Y|X)$ in the bivariate normal model as functions of $\rho$.}
\label{fig:normal-coes-scoes-rho}
\end{figure}
\par
Due to Corollary~\ref{cor:2}(\ref{item:cor.2.a}) we already know that 
$\coesab$ is increasing in $\rho$ for all $\alpha$ and $\beta$. 
The special case $\alpha=\beta$ is illustrated in 
Figure~\ref{fig:normal-coes-scoes-rho}, which shows that $\ABcoes$ is 
not increasing in $\rho$. 
Due to the light tail of the normal distribution, 
these plots are similar to those of $\covar$ and $\ABcovar$ in 
Figure~\ref{fig:normal-covar-scovar-rho}.  
A closer look at~\eqref{eq:019a} confirms that the non-monotonicity 
of $\ABcoes$ in $\rho$ is inherited from $\ABcovar$. 
Thus the best possible extension to Conditional Expected Shortfall 
based on $\ABcovar$ still fails to reflect dependence properly.
%
%
%
\subsection{Bivariate $t$ distribution}\label{subsec:ellipt-t}
The next example we consider is the bivariate $t$ distribution, which 
is elliptical, but heavy-tailed. The comparison follows the same scheme 
as in the previous section. 
A bivariate $t$ distributed random vector with $\nu>0$ degrees of freedom 
(bivariate $t(\nu)$) can be obtained as follows:
\[
(X,Y) := (\muX,\muY) + \sqrt{\frac{\nu}{W}}\robrfl{\Xtilde,\Ytilde}, 
\]
where $(\Xtilde,\Ytilde)\sim \Ncal(0,\Sigma)$ 
and $W\sim\chi^2(\nu)$, independent of $(\Xtilde,\Ytilde)$.
The parameters $\muX,\muY\in\R$ specify the location of $(X,Y)$. 
For simplicity, we consider a centred model with $\muX=\muY=0$. 
\par
It is well known that the bivariate $t$ distribution is elliptical 
with ellipticity matrix $\Sigma$. 
The corresponding sample clouds have an elliptical shape 
(cf.\ Figure~\ref{fig:simul-t}).
The second moments of $X$ and $Y$ are finite for $\nu>2$, 
and in this case  
the correlation between $X$ and $Y$ is equal to $\rho$. 
The role of $\rho$ is the same as for all elliptical models: 
larger values of $\rho$ increase association between 
large values of $X$ and $Y$.
Analytic expressions for $\ABcovar$ or $\covar$ are not feasible 
in this model, so that computations have to be carried 
out numerically. 
\subsubsection*{Monotonicity in $\rho$}
The behaviour of $\ABcovar$ and $\covar$ as functions of the correlation 
parameter $\rho$ is shown in 
Figure~\ref{fig:T-covar-scovar}. 
Similarly to the Gaussian case, $\covar$ is increasing in $\rho$ due 
to Theorem~\ref{thm:4}(\ref{item:thm.4.a}), whereas $\ABcovar$ is not.
Moreover, the relative distance between $\ABcovar$ and $\covar$ 
(as it could be quantified by the ratio $\covar/\ABcovar$) is larger than 
in the Gaussian case. A possible explanation to this effect could be 
the heavy tail of the $t(3)$ distribution.  
\par
\begin{figure} 
\centering
\includegraphics[width=0.8\textwidth]{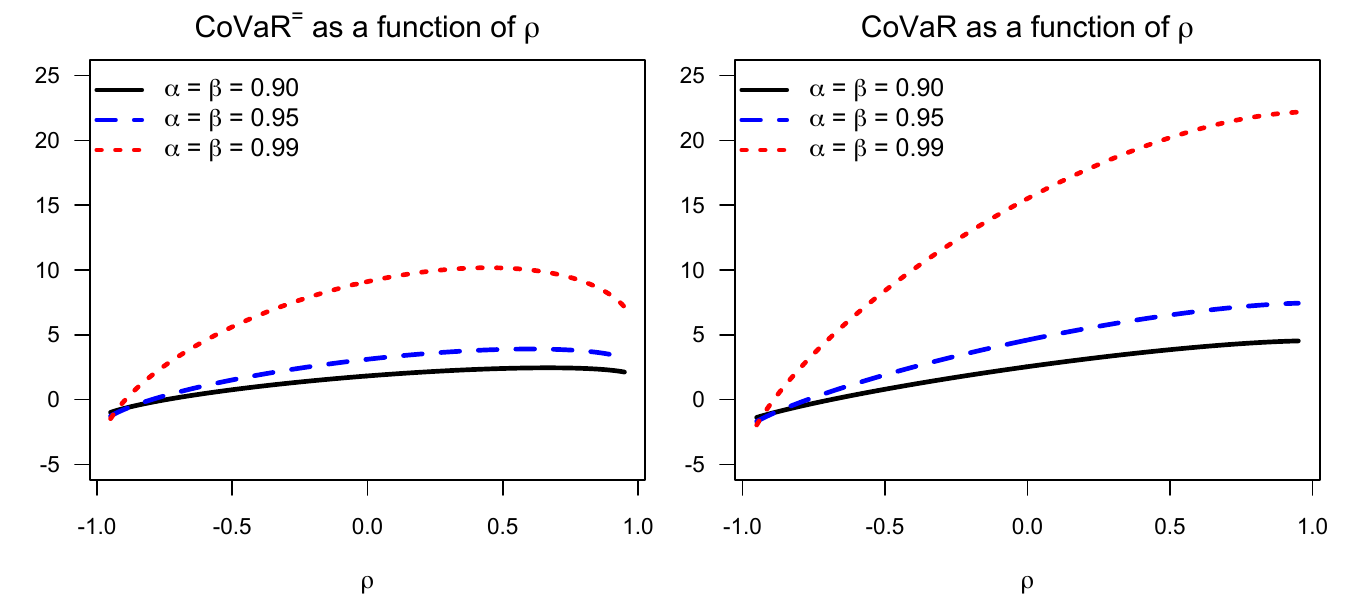}
\caption{Bivariate $t(3)$ distribution: $\ABcovara$ and $\covara$ as functions of the correlation parameter $\rho$.}%
\label{fig:T-covar-scovar}%
\end{figure}
\par
\subsubsection*{Normalized values of $\covar$ and $\ABcovar$}
Figure~\ref{fig:ratios-t} shows 
the ratios $\ABcovara(Y|X)/\var_\alpha(Y)$ 
and $\covara(Y|X)/ \var_\alpha(X)$ as functions of $\alpha$ for  
selected values of $\rho$. 
This comparison is analogous to 
Figure~\ref{fig:ratios-normal} in the Gaussian case. 
Similarly to 
the Gaussian case, the ordering of $\ABcovara/\var_\alpha$ with respect to the 
dependence parameter $\rho$ or $\theta$ is inconsistent, whereas 
the ratios $\covara/ \var_\alpha$ are ordered correctly for all 
$\alpha$: the line for the largest $\rho$ is entirely above the line for 
the second largest $\rho$, etc. 
In contrast to the Gaussian case, these ratios are increasing in $\alpha$. 
This could be explained by the heavy tail of the $t(3)$ distribution or 
by the positive tail dependence in the bivariate $t$ model. 
\par
\begin{figure}
\centering
\includegraphics[width=0.8\textwidth]{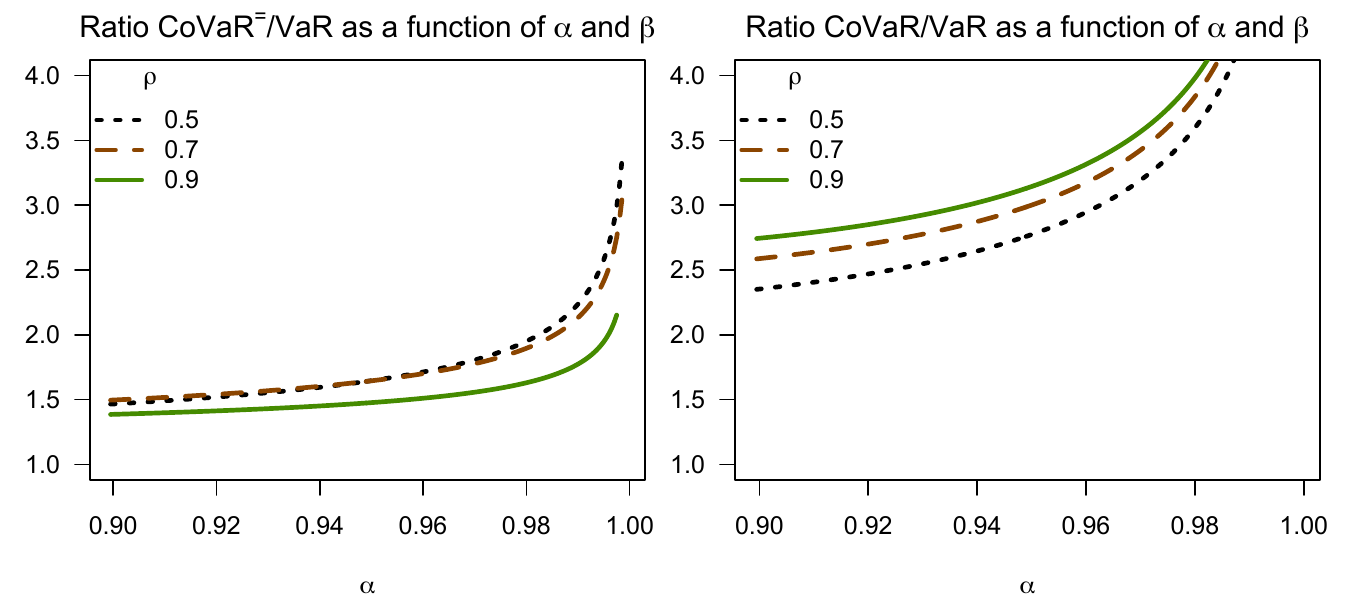}
\caption{Bivariate $t(3)$ distribution with $\mu_Y=0$: Ordering of the ratios $\ABcovara(Y|X)/\var_\alpha(Y)$ and $\covara(Y|X)/\var_\alpha(Y)$ for different $\alpha$.}
\label{fig:ratios-t}
\end{figure}
\par
\subsubsection*{Backtesting and violation rates}
The backtesting study was implemented analogously to  
the bivariate Gaussian example. 
The results are shown in Table~\ref{table:t-table}, 
and they go in line with those from the Gaussian case. 
While $\covar$ -- again, by construction -- has a violation rate 
close to $1-\beta$,  
the violation rates of $\ABcovar$ are significantly higher 
and increase in $\rho$. 
Going up to $36\%$ for $\rho=0.9$, the violation rates for 
$\covar$ are even higher than in the Gaussian model. 
\par
\begin{table}
\centering
\begin{tabular}{c c c c c c} 
\hline\hline
  \  & $\rho=0$ & $\rho=0.2$ & $\rho=0.5$ & $\rho=0.7$ & $\rho=0.9$ \\[0.5ex]
\hline
$\covar^{=}_{0.95, 0.95}(Y|X)$ & 0.1017 & 0.1213 & 0.1659 & 0.2202 & 0.3638 \\
$\covar^{=}_{0.99, 0.99}(Y|X)$ & 0.0358 & 0.0433 & 0.0643 & 0.0939 & 0.1909 \\
$\covar^{=}_{0.95, 0.99}(Y|X)$ & 0.0341 & 0.0429 & 0.0640 & 0.0944 & 0.1954 \\
$\covar^{=}_{0.99, 0.95}(Y|X)$ & 0.1036 & 0.1229 & 0.1658 & 0.2184 & 0.3546 \\[0.9ex]
$\covar_{0.95, 0.95}(Y|X)$ & 0.0497 & 0.0500 & 0.0499 & 0.0506 & 0.0504 \\
$\covar_{0.99, 0.99}(Y|X)$ & 0.0103 & 0.0099 & 0.0104 & 0.0105 & 0.0103 \\
$\covar_{0.95, 0.99}(Y|X)$ & 0.0100 & 0.0099 & 0.0100 & 0.0102 & 0.0101 \\
$\covar_{0.99, 0.95}(Y|X)$ & 0.0501 & 0.0493 & 0.0499 & 0.0508 & 0.0507 \\
\hline
\end{tabular}
\caption{Violation rates in the bivariate $t(3)$ case. Monte Carlo backtesting with $n=10^{7}$ and $\alpha,\beta\in\cubr{0.95,0.99}$.%
} 
\label{table:t-table}
\end{table}
\par
The corresponding sample plots with lines marking $\var_\alpha(X)$, 
$\ABcovara(Y|X)$, and $\covara(Y|X)$  
are shown in Figure~\ref{fig:simul-t}.
Similarly to Figure~\ref{fig:simul-normal-1}, these graphics demonstrate 
how increasing dependence parameter $\rho$ changes the shape 
of the corresponding sample clouds and increases the numbers of joint 
excesses. 
\par
\begin{figure}
\centering
\includegraphics[width=\textwidth]{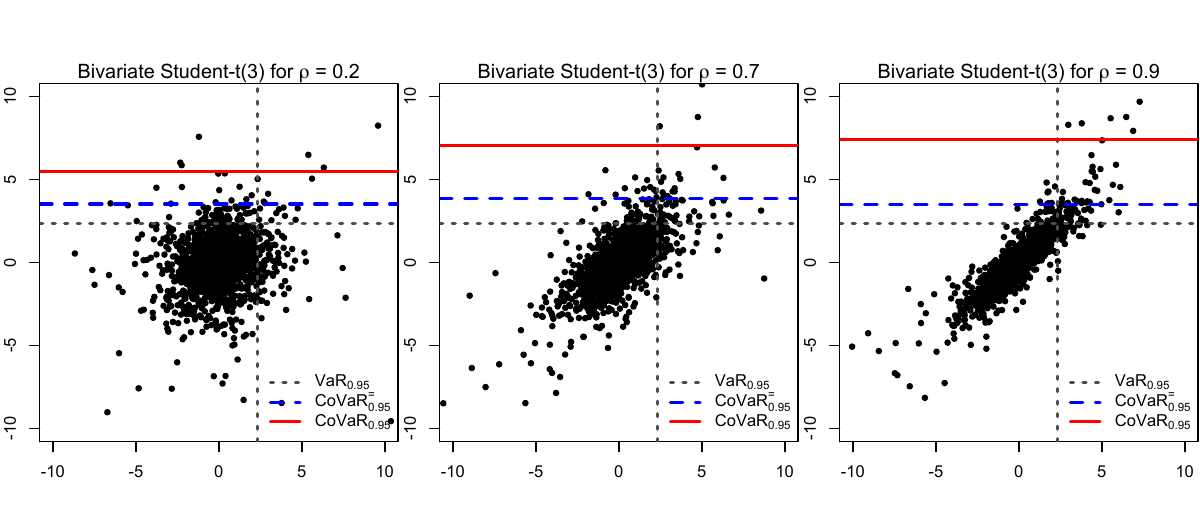}
\caption{Bivariate $t(3)$ samples (size $n=2000$)
and the joint excess regions in the backtesting experiment for $\alpha=\beta=0.95$.}
\label{fig:simul-t}
\end{figure}
\subsubsection*{Risk contribution measures $\dcovar$ and $\dmcovar$} 
The comparison of $\dcovar$ and $\dmcovar$ is shown in 
Figure~\ref{fig:T-dcovar}. 
The graphics demonstrate clearly how these $\ABcovar$ 
based risk contribution 
measures inherit the inconsistency of $\ABcovar$. 
Both $\dcovar$ and $\dmcovar$ fail to be increasing with respect 
to the dependence parameter $\rho$, and the shapes 
of the corresponding curves are similar to those of $\ABcovar$ in 
Figure~\ref{fig:T-covar-scovar}.  
Although $\dmcovar$ is slightly better behaved than $\dcovar$, 
it is still strongly inconsistent with respect to $\rho$. 
In particular, this example demonstrates that the monotonicity of 
$\dmcovar$ with respect to $\rho$ in the Gaussian case 
is a special property of the bivariate Gaussian model, 
so that the advantage of $\dmcovar$ over $\dcovar$ is quite limited. 
\par
\begin{figure}
\centering
\includegraphics[width=0.8\textwidth]{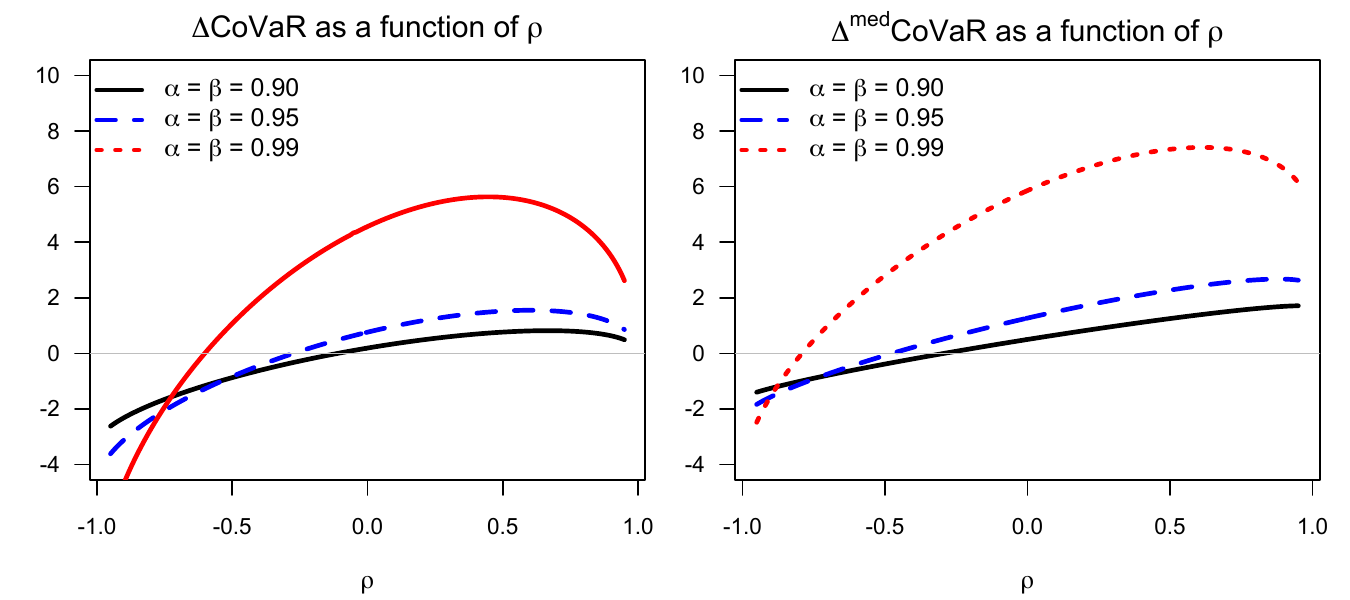}
\caption{$\dcovara$ and $\dmcovara$ as functions of $\rho$ in the bivariate $t(3)$ model.}
\label{fig:T-dcovar}
\end{figure}
\par
\subsubsection*{Extension from $\covar$ to $\coes$}
\par
\begin{figure}
\centering
\includegraphics[width=0.8\textwidth]{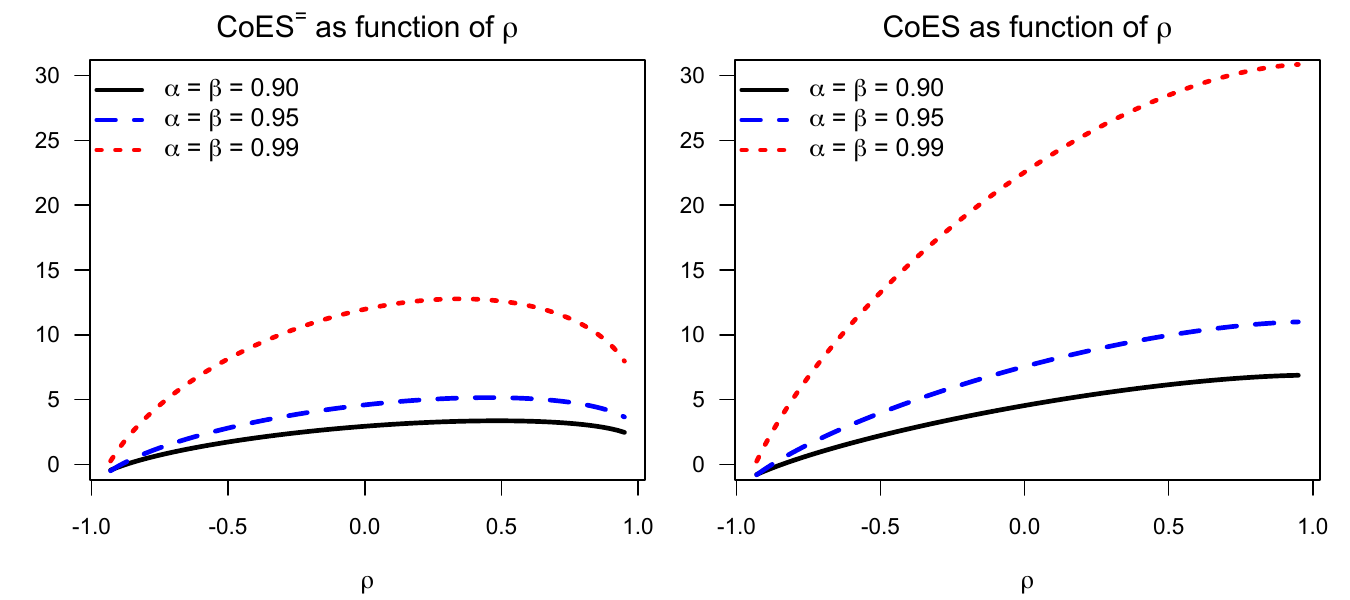}
\caption{$\ABcoesa (Y|X)$ and $\coesa(Y|X)$ in the bivariate $t(3)$ model as functions of $\rho$.}
\label{fig:T-coes}
\end{figure}
\par
The comparison of $\coes$ vs.\ $\ABcoes$ is shown in Figure~\ref{fig:T-coes}.  
The monotonicity or non-monotonicity in $\rho$ is again inherited from $\covar$ 
or $\ABcovar$. See also Corollary~\ref{cor:2}(\ref{item:cor.2.a}).
\subsection{Gumbel copula with $t$ margins}\label{subsec:gumbel-t}
The last model we consider here is obtained by endowing a 
bivariate Gumbel copula (cf.~\eqref{eq:023}) with $t$ margins. 
Thus it has the same heavy-tailed margins as the previous example, 
but a different dependence structure. 
An illustration of the sample clouds generated from this 
distribution is given in Figure~\ref{fig:simul-gumbel-1}.
\par
On the qualitative level, all comparison results obtained in this case are 
similar to the bivariate $t$ model, so that a brief overview is fully 
sufficient: 
\begin{itemize}
\item 
Corollary~\ref{cor:1} guarantees that $\covarab$ is increasing 
with respect to the dependence parameter $\theta$, whereas 
$\ABcovarab$ fails to be increasing if dependence is at its largest 
(see Figure~\ref{fig:gumbel-t-covar-scovar} for the case $\alpha=\beta$).
The strongest decay of $\ABcovar$ takes place for $\theta\in(1.5,2)$ 
and slows down for $\theta>2$. On the other hand, $\covara$ is almost 
constant for $\theta>2$. 
It seems that for $\theta>2$ the joint distribution of large values 
of $(X,Y)$ is almost comonotonic, 
so that there is no much change after $\theta$ exceeds $2$.
\item
The ratios $\covara(Y|X)/\var_\alpha(Y)$ are ordered correctly 
with respect to $\theta$, whereas the ratios 
$\ABcovara(Y|X)/\var_\alpha(Y)$ are not 
(see Figure~\ref{fig:ratios-gumbel}).
\item
The violation rates for $\ABcovarab$ in a simulated backtesting study 
are significantly larger than $1-\beta$, going up to $40\%$ 
for $\alpha=\beta=0.95$ and $\theta=3$ (cf.\ Table~\ref{table:gumbel-table} 
and Figure~\ref{fig:simul-gumbel-1}). 
This is even more than in the bivariate $t$ case. 
\item
Both $\dcovar$ and $\dmcovar$ fail to be increasing in $\theta$ 
(Figure~\ref{fig:gumbel-dcovar}).
\item Again, $\coes$ is increasing in $\theta$ and $\ABcoes$ is not; 
see Corollary~\ref{cor:2}(\ref{item:cor.2.a}) and Figure~\ref{fig:gumbel-t-coes}.
\end{itemize}
\par
\begin{figure}
\centering
\includegraphics[width=0.8\textwidth]{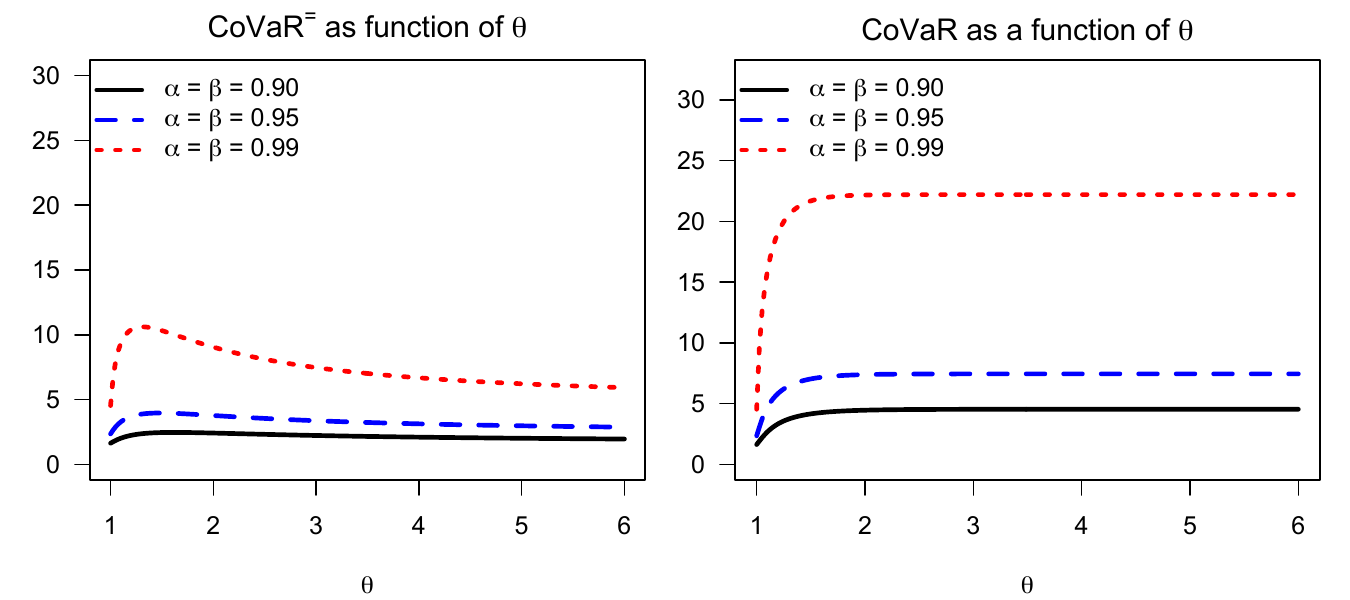}
\caption{Gumbel copula $t(3)$ margins: $\ABcovara (Y|X)$ and $\covara(Y|X)$ as functions of $\theta$.}%
\label{fig:gumbel-t-covar-scovar}%
\end{figure}
\par
\begin{figure}
\centering
\includegraphics[width=0.8\textwidth]{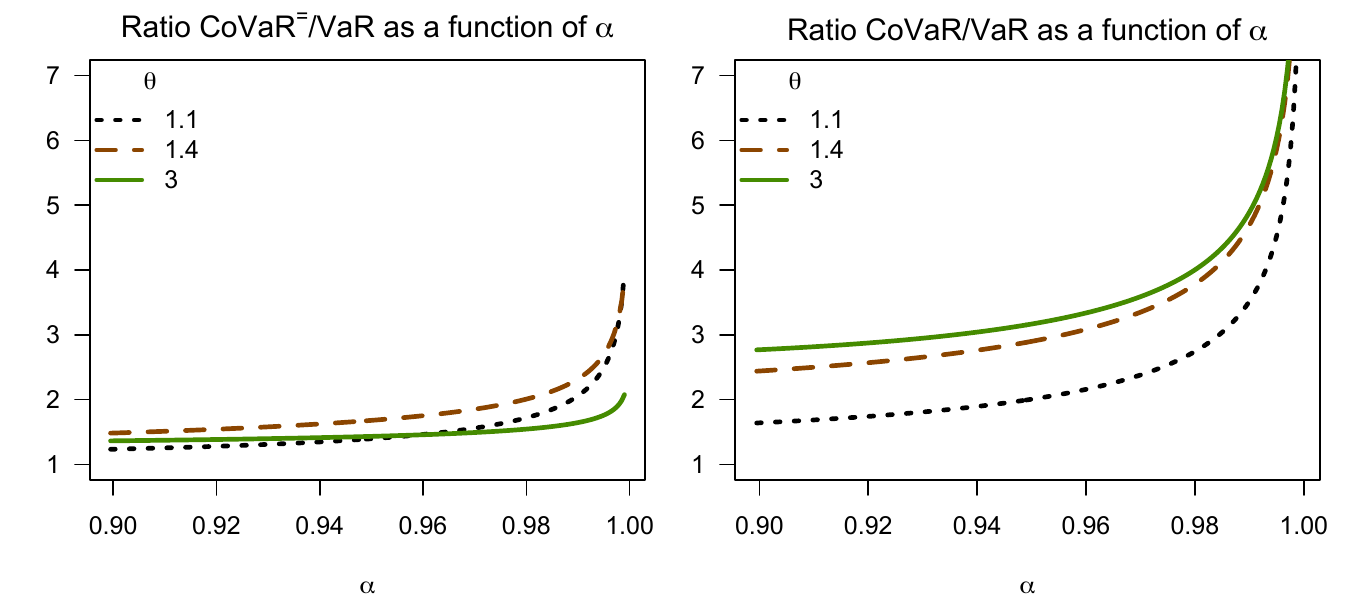}
\caption{Gumbel copula with $t(3)$ margins: Ordering of the ratios $\ABcovara(Y|X)/\var_\alpha(Y)$ and $\covara(Y|X)/\var_\alpha(Y)$ for different $\alpha$.}
\label{fig:ratios-gumbel}
\end{figure}
\par
\begin{table}[ht]
\centering 
\begin{tabular}{c c c c c c c} 
\hline\hline
\  & $ \theta = 1$ & $ \theta = 1.1 $  & $ \theta = 1.2 $ & $ \theta = 1.5 $ & $ \theta = 2 $ & $ \theta = 3 $ \\ [0.5ex] 
\hline 
$\covar^{=}_{0.95, 0.95}(Y|X)$ & 0.0498 & 0.0982 & 0.1282 & 0.1911 & 0.2771 & 0.4090 \\ 
$\covar^{=}_{0.99, 0.99}(Y|X)$ & 0.0101 & 0.0346 & 0.0461 & 0.0752 & 0.1321 & 0.2423 \\
$\covar^{=}_{0.95, 0.99}(Y|X)$ & 0.0098 & 0.0309 & 0.0434 & 0.0754 & 0.1319 & 0.2450 \\
$\covar^{=}_{0.99, 0.95}(Y|X)$ & 0.0500 & 0.1050 & 0.1335 & 0.1916 & 0.2745 & 0.4043 \\[0.9ex]
$\covar_{0.95, 0.95}(Y|X)$ & 0.0498 & 0.0494 & 0.0503 & 0.0498 & 0.0501 & 0.0502 \\
$\covar_{0.99, 0.99}(Y|X)$ & 0.0101 & 0.0099 & 0.0101 & 0.0102 & 0.0100 & 0.0097 \\
$\covar_{0.95, 0.99}(Y|X)$ & 0.0098 & 0.0099 & 0.0100 & 0.0099 & 0.0100 & 0.0098 \\
$\covar_{0.99, 0.95}(Y|X)$ & 0.0500 & 0.0497 & 0.0499 & 0.0492 & 0.0503 & 0.0492 \\
\hline 
\end{tabular}
\caption{Violation rates for the Gumbel copula with $t(3)$ margins: Monte Carlo backtesting with $n=10^{7}$ 
and $\alpha,\beta\in\cubr{0.95,0.99}$.%
}
\label{table:gumbel-table} 
\end{table}
\par
\begin{figure}
\centering
\includegraphics[width=\textwidth]{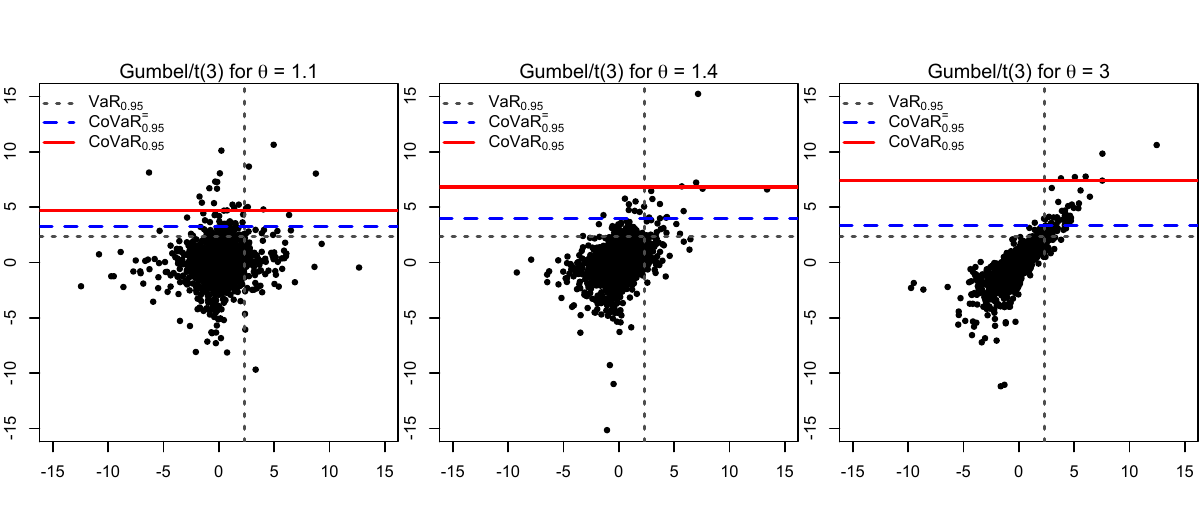}
\caption{Gumbel copula with $t(3)$ margins: simulated samples (size $n=2000$)
and the joint excess regions in the backtesting experiment for $\alpha=\beta=0.95$.}%
\label{fig:simul-gumbel-1}
\end{figure}
\par
\begin{figure}
\centering
\includegraphics[width=0.8\textwidth]{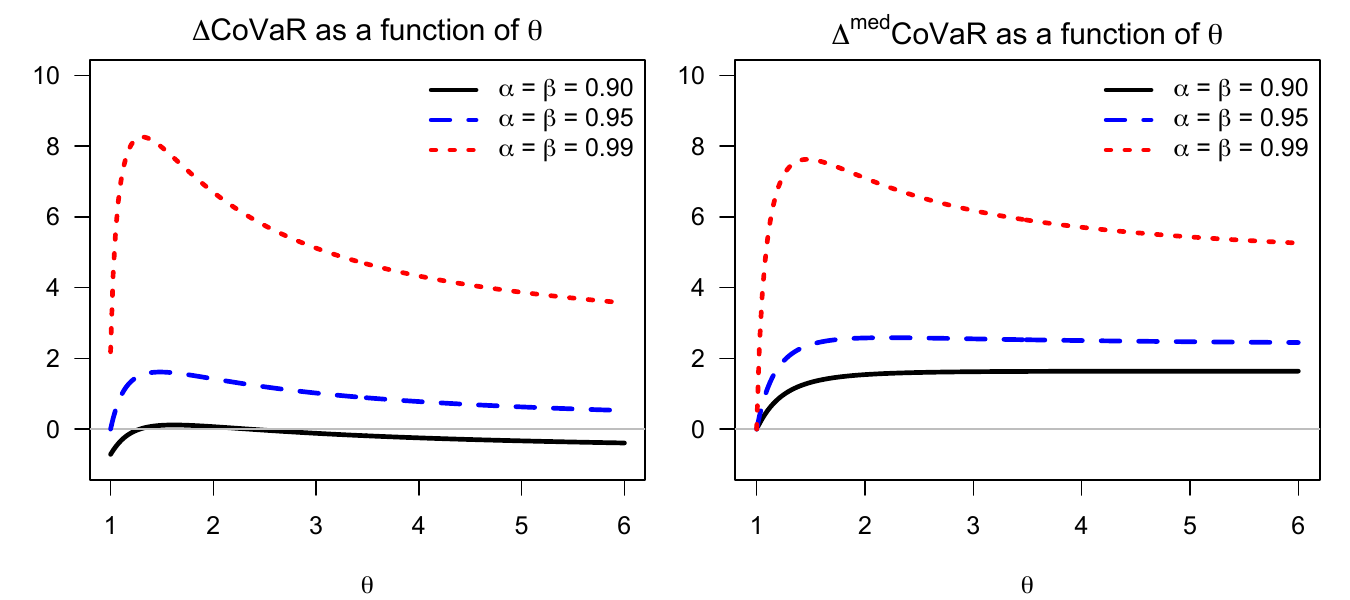}
\caption{Gumbel copula with $t(3)$ margins: $\dcovara$ and $\dmcovara$ as functions of $\theta$.}
\label{fig:gumbel-dcovar}
\end{figure}
\par
\begin{figure}
\centering
\includegraphics[width=0.8\textwidth]{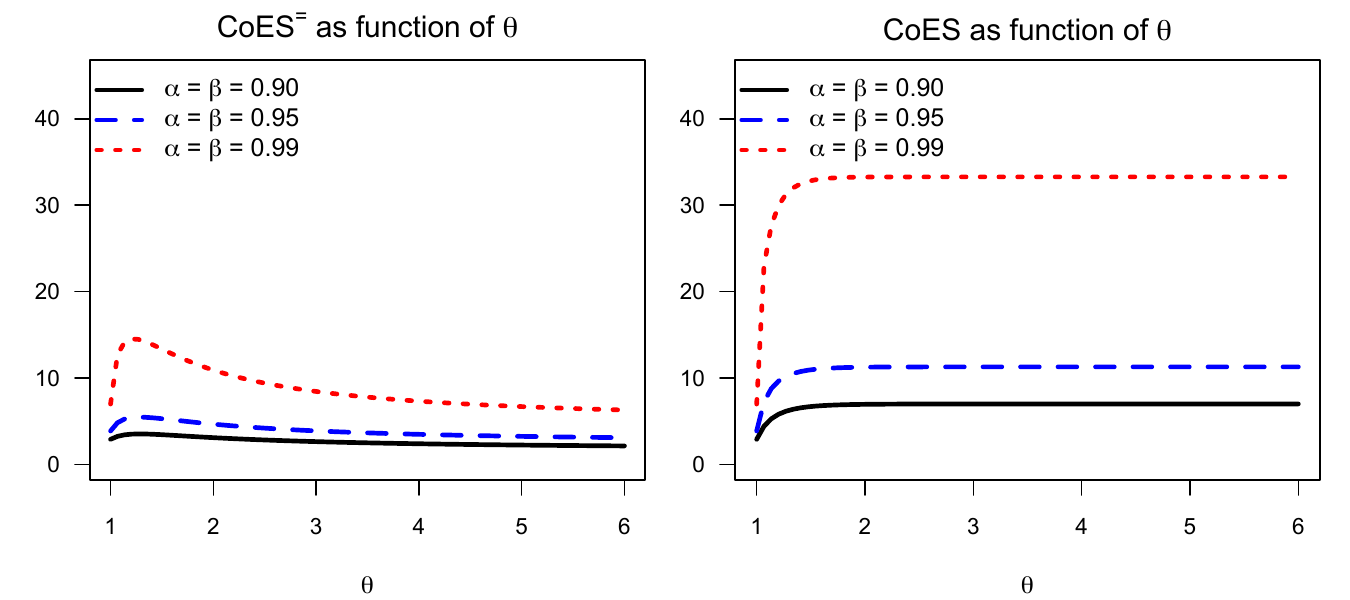}
\caption{Gumbel copula with $t(3)$ margins: $\ABcoesa (Y|X)$ and $\coesa(Y|X)$ as functions of $\theta$.}
\label{fig:gumbel-t-coes}%
\end{figure}
\par
\FloatBarrier 
%
%
%

%% file: conclusions.tex
\section{Conclusions} \label{sec:conclusions}
The present paper demonstrates that the alternative 
definition of Conditional Value-at-Risk
proposed in~\cite{Girardi/Erguen:2011,Klyman:2011}  (here $\covar$) 
gives a much more consistent 
response to dependence than the original definition
used in~\cite{covar08, covar09, covar11}  (here $\ABcovar$).
\par
The general results in Section~\ref{sec:general} show that the 
monotonicity of $\covarab(Y|X)$ with respect to dependence parameters is 
related to the concordance ordering of bivariate distributions or copulas.
This gives the notion of $\covar$ based on the stress scenario 
$\cubr{X\ge\var_\alpha(X)}$ a solid mathematical fundament. 
On the other hand, comparative studies in Section~\ref{sec:examples} 
show that conditioning on $\cubr{X=\var_\alpha(X)}$ makes $\ABcovar$ 
and its derivatives unable to detect systemic risk where it is most 
pronounced. 
Related counterexamples include several popular models, 
in particular the very basic bivariate normal case.  
\par
Based on these results, we claim that, if Conditional Value-at-Risk 
of an institution (or system) $Y$ 
related to a stress scenario for another institution 
$X$ should enter financial regulation, 
then it should use conditioning on $\cubr{X\ge\var_\alpha(X)}$. 
This kind of stress scenario has a much more meaningful practical  
interpretation than the highly selective and over-optimistic 
scenario $\cubr{X=\var_\alpha(X)}$.  
Conditioning on $\cubr{X\ge\var_\alpha(X)}$ also makes $\covar$ 
more similar to the systemic risk 
measures proposed in \cite{Goodhart/Segoviano:2008,Acharya/Pedersen/Philippon/Richardson:2010,Zhou:2010,Huang/Zhou/Zhu:2011}. 
\par
The question how to define risk contribution measures based on stress 
events to the financial system is currently open. Besides $\covar$, 
$\coes$ with proper conditioning may also be an option. 
The advantage of $\coes$ over $\covar$ is its coherency. 
In the case $\var$ vs. $\es$, this point has gained new interest from 
the regulators \cite{bcbs219, Gauthier/Lehar/Souissi:2012}. 
\par
In some sense, $\ABcovar$ repeats two times the design error
that is responsible for the non-coherency of $\var$.
In the first step, it follows the $\var$ paradigm and thus 
favours a single conditional quantile of $Y$ over an average of such quantiles.
In the second step, it favours the most benign outcome of $X$ 
in a state of stress over considering the full range of possible 
values in this case. 
Financial regulation based on $\ABcovar$
has a strong potential to introduce additional instability, to set 
wrong incentives, and to create opportunities for regulatory arbitrage. 
\par
Another argument supporting $\coes$ is that it is particularly suitable 
for stress testing. In a system with several factors $X_1,\ldots,X_d$, 
the numbers 
$\coes_{\alpha_i, \beta}(Y|X_i)$ 
describe the influence of the different $X_i$ on $Y$. 
Assigning relative weights $w_i$ to the scenarios 
$X_i\ge\var_{\alpha_i}(X_i)$ and taking the weighted sum 
\begin{align} \label{eq:009}
\sum_{i=1}^d w_i \coes_{\alpha_i, \beta}(Y|X_i),
\end{align}
one always obtains a sub-additive risk measure. 
If the weights $w_i$ sum up to $1$, the resulting risk measure 
is coherent in the sense of \cite{Artzner/Delbaen/Eber/Heath:1999}.
The choice of the weights $w_i$ or of the confidence levels $\alpha_i$ 
may change over time, 
incorporating the newest information about the health of the institutions 
$X_1,\ldots,X_d$. 
\par
To make the weighted risk measure~\eqref{eq:009} even more meaningful, 
one could modify it by implementing not only the single risk 
factor excesses $X_i\ge\var_{\alpha_i}(X_i)$, but also the joint ones. 
Consistent choice of the corresponding weights can be derived by 
methods presented in \cite{Rebonato:2010}.  
A detailed discussion of this goes beyond the scope of 
the present paper and would also require additional mathematical research. 
\par
Motivated by the recent financial crisis and the following discussions 
on appropriate reforms in financial regulation, systemic risk measurement 
has become a vivid topic in economics and econometrics. 
Our results show that some important contributions are 
also to be made in related mathematical fields, 
including probability and statistics.
In particular, the dependence consistency or, say, dependence coherency of 
systemic risk indicators is a novel problem area that needs further 
study. The present paper provides first examples and counter-examples 
for compatibility of systemic risk indicators with the concordance 
order. 
%
The questions for general characterizations or representations of 
functionals with this property 
are currently open. 
\par 
In addition to dependence consistency, implementation of 
systemic risk measures in practice obviously needs  
estimation methods. The estimation of $\covar$ in GARCH models is discussed 
in~\cite{Girardi/Erguen:2011}. 
As non-parametric estimation of rare events would needs a lot of data, 
methods from Extreme Value Theory may be used to extrapolate the 
rear events from a larger number of data points. 
A similar approach for conditional default probabilities is 
pursued in~\cite{Zhou:2010}. 
%
%
%
%

%% file: acknowledgements.tex
\section*{Acknowledgements}
The authors would like to thank Paul Embrechts for several fruitful discussions 
related to this paper. Georg Mainik thanks RiskLab, ETH Zürich, for financial 
support. 
%
%
%
%
%